\documentclass[12pt]{article}
\usepackage{makeidx}
\usepackage{amsmath}
\usepackage{amsfonts}
\usepackage{amssymb}
\usepackage{epsfig}
\usepackage[cp1251]{inputenc}
\usepackage{amsthm}

\newtheorem{Th}{Theorem}
\newtheorem{Lem}{Lemma}
\newtheorem{Rem}{Remark}
\newtheorem{prop}{Proposition}
\newtheorem{Cor}{Corollary}

\def\1{{\bf 1}}

\usepackage{algorithm}
\usepackage{algpseudocode}
\usepackage{amsmath}
\usepackage{xcolor}
\usepackage{listings}
\usepackage{alltt}

\begin{document}

\title{Quantum accessible information and classical entropy inequalities}
\author{A. S. Holevo, A. V. Utkin \\
Steklov Mathematical Institute, RAS, Moscow, Russia}
\date{}
\maketitle

\begin{abstract}
Computing accessible information for an ensemble of quantum states is a
basic problem in quantum information theory. We show that the recently
obtained optimality criterion (A.S. Holevo, Lobachevskii J. Math., \textbf{43}:7 (2022),
1646-1650), when applied to specific ensembles
of states leads to nontrivial tight entropy inequalities that are discrete
relatives of the famous log-Sobolev inequality. In this light, the
hypothesis of globally information-optimal measurement for an ensemble of
equiangular equiprobable states (quantum pyramids) (B.-G. Englert and J. \v{R}eh\'{a}\v{c}ek, J. Mod. Optics \textbf{57 }N3 (2010)
218-226)  is
reconsidered and the corresponding entropy inequalities are proposed. Via
the optimality criterion, this suggests also an approach to the proof of the
conjectures concerning globally information-optimal observables for quantum
pyramids.

\end{abstract}


\section{\protect\bigskip Introduction}

The problem of accessible information -- the maximal information that can be
extracted from measurements over an ensemble of quantum states -- is a basic
problem in quantum information and communication theory, coming back to its
origin, see e.g. \cite{jma1}, \cite{obzor}, \cite{jozsa}, \cite{hkron}, and
important in particular for applications in quantum cryptography.
The problem was solved in some special (notably, symmetric) cases but in
general remains open. In section (\ref{aopt}) of this paper we describe a
general approach following \cite{ljm} which gives a criterion for optimality
and ties up the problem with \textit{entropy inequalities} -- the tight
lower bounds for the classical Shannon entropy. It turns out that the
optimality criterion when applied to specific ensembles of states leads to
new nontrivial entropy inequalities that are discrete \textquotedblleft
distant relatives\textquotedblright\ of the famous log-Sobolev inequality,
as explained in section \ref{sec:two}.

In this light, the conjecture of an information-optimal measurement for an
ensemble of equiprobable equiangular states -- \textit{quantum pyramid},
formulated and numerically substantiated by Englert and \v{R}eh\'{a}\v{c}ek
in \cite{eng} -- is reconsidered in section \ref{sec:multi} of the present
paper. In the case of \textit{acute} pyramids the corresponding parametric
family of entropy inequalities (\ref{basic}) is proposed in subsections \ref%
{conj}-\ref{acute}. The inequalities suggest the new tight lower bounds for
the Shannon entropy of a discrete probability distribution as compared e.g.
to the known bound from \cite{moser}, theorem 3.11.

In the subsection \ref{obtus} the case of \textit{obtuse} pyramids is
considered to get further insight and to formulate the hypotheses concerning
the corresponding entropy inequalities (\ref{basic2}), (\ref{ineqz}). Our
approach allows also to derive the nonlinear equations (\ref{eqt1}) (resp. (%
\ref{eqt2})) serving to determine the critical values of the parameter $t$
of the optimal observable corresponding to transition from moderately to
strongly acute (resp, obtuse) pyramids.

In the subsection \ref{flat} we derive the tight entropy inequality (\ref%
{ineqw}) implying the global optimality of the hypothetical observable from
\cite{eng} (proposition \ref{prop1} and corollary \ref{cor3}) in the case of
\textit{flat} pyramids. The findings in subsections \ref{obtus}, \ref{flat}
also emphasize that for obtuse and flat pyramids in the dimensions $m\geq 7$
the sharp \textquotedblleft square root\textquotedblright\ observable (SRM)
outperforms the \textquotedblleft unambiguous
discrimination\textquotedblright\ observable (UDM) and tentatively is the
optimal one.

The proposed entropy inequalities may be more convenient for numerical check
than the direct optimization, however proving them analytically is a
challenge. In section \ref{prrrofs}, the suggestions for proofs proposed by
the second author as a response to the hypotheses (\ref{basic}), (\ref%
{basic2}), (\ref{ineqz}) put forward by the first author, are presented. Via
the optimality criterion of section \ref{aopt}, this suggests also an
approach to the proof of the conjecture concerning globally
information-optimal observables for quantum pyramids formulated and
numerically substantiated in \cite{eng}.

\section{\protect\bigskip Accessible information: a criterion for optimality}

\label{aopt}

Let $\mathcal{H}$ be a finite dimensional Hilbert space, $\mathfrak{S}=%
\mathfrak{S}(\mathcal{H})$ the convex set of quantum states (density
operators on $\mathcal{H}$). \ \textit{Ensemble} of quantum states is a
collection $\mathcal{E}=\left\{ \pi _{j},\rho _{j}\right\} _{j=1,\dots ,m}$
where $\rho _{j}\in \mathfrak{S,}$ and $\pi =\left\{ \pi _{j}\right\} $ is a
probability distribution. Let $\mathcal{M}=\left\{ M_{k}\right\} _{k=1,\dots
,n}$ be an observable i.e. a collection of operators on $\mathcal{H}$ such
that $M_{k}\geq 0,\,\sum_{k=1}^{n}M_{k}=I$ (unit operator). The joint
distribution of \textquotedblleft input\textquotedblright\ $j$ and
\textquotedblleft output\textquotedblright\ $k$ is then%
\begin{equation*}
p_{jk}=\pi _{j}\mathrm{Tr\,}\,\rho _{j}M_{k}
\end{equation*}%
and the Shannon mutual information between $j$ and $k$ is (throughout the
paper $\log $ denotes the binary logarithm)
\begin{equation*}
I(\mathcal{E},\mathcal{M})=\sum_{j,k}p_{jk}\log \frac{p_{jk}}{\pi _{j}p_{k}},
\end{equation*}%
with
\begin{equation*}
p_{k}=\sum_{j}p_{jk}=\mathrm{Tr\,}\,\overline{\rho }M_{k},
\end{equation*}%
where $\overline{\rho }=\sum_{j}\pi _{j}\rho _{j}$ is the \textit{average
state} of the ensemble $\mathcal{E}.$

\textit{Accessible information} of the ensemble is defined as
\begin{equation}
A(\mathcal{E})=\sup_{\mathcal{M}}I(\mathcal{E},\mathcal{M}).  \label{ai}
\end{equation}%
It was shown \cite{dav} that the supremum is attained on an observable of
the form $M_{k}=\left\vert \varphi _{k}\rangle \langle \varphi
_{k}\right\vert ,\,k=1,\dots ,n$ with $d\leq n\leq d^{2}$ and linearly
independent $M_{k}$ (if the matrices of $\rho _{j}$ in some basis are real
then the number of components is reduced to $\frac{d(d+1)}{2}$). In
particular, the supremum can be replaced with the maximum.

From now on we assume that the \textit{average state} $\overline{\rho }$
\textit{is nondegenerate}. Introducing the \textit{dual} ensemble $\mathcal{E%
}^{\prime }=\left\{ p_{k},\sigma _{k}\right\} $ with
\begin{equation*}
\sigma _{k}=p_{k}^{-1}\overline{\rho }^{1/2}M_{k}\overline{\rho }%
^{1/2}=\left\vert \phi _{k}\rangle \langle \phi _{k}\right\vert ,
\end{equation*}%
where
\begin{equation*}
|\phi _{k}\rangle =p_{k}^{-1/2}\overline{\rho }^{1/2}|\varphi _{k}\rangle ,
\end{equation*}%
and observable $\mathcal{M}^{\prime }$ with the components
\begin{equation}  \label{mprim}
M_{j}^{\prime }=\overline{\rho }^{-1/2}\left( \pi _{j}\rho _{j}\right)
\overline{\rho }^{-1/2}
\end{equation}%
we have the same average state
\begin{equation*}
\overline{\sigma }=\sum_{k}p_{k}\sigma _{k}=\overline{\rho }
\end{equation*}%
and the new expression for accessible information:
\begin{equation}
A(\mathcal{E})=\max_{\mathcal{E}^{\prime }:\overline{\sigma }=\overline{\rho
}}I(\mathcal{E}^{\prime },\mathcal{M}^{\prime })  \label{ae1}
\end{equation}%
with%
\begin{eqnarray*}
I(\mathcal{E}^{\prime },\mathcal{M}^{\prime }) &=&I(\mathcal{E},\mathcal{M}%
)=-\sum_{j}\pi _{j}\log \pi _{j}+\sum_{k,j}p_{k}p(j|k)\log p(j|k) \\
&=&H(\pi )-\sum_{k}\mathrm{Tr\,}\,\left( p_{k}\sigma _{k}\right) K\left(
\sigma _{k}\right) ,
\end{eqnarray*}%
where $p(j|k)=\mathrm{Tr\,}\,\sigma _{k}M_{j}^{\prime },$%
\begin{equation}
K\left( \sigma \right) =-\sum_{j}M_{j}^{\prime }\log \mathrm{Tr\,}\,\sigma
M_{j}^{\prime }.  \label{K}
\end{equation}%
Therefore%
\begin{equation}
A(\mathcal{E})=H(\pi )-\min_{\mathcal{E}^{\prime }:\overline{\sigma }=%
\overline{\rho }}\sum_{k}\mathrm{Tr\,}\,\left( p_{k}\sigma _{k}\right)
K\left( \sigma _{k}\right)  \label{ae2}
\end{equation}

\begin{Rem}
On the other hand, (\ref{ae1}) is also the value of $\overline{\rho }-$%
\textit{constrained classical capacity} $C(\mathcal{M}^{\prime },\overline{%
\rho })$ of the measurement channel corresponding to the observable $%
\mathcal{M}^{\prime }$ \cite{ljm1}. If $H$ is the Hamiltonian and $E$ is the
energy limit then the energy-constrained capacity is
\begin{equation*}
C(\mathcal{M}^{\prime },E)=\max_{\overline{\rho }:\mathrm{Tr\,}\overline{%
\rho }H\leq E}C(\mathcal{M}^{\prime },\overline{\rho })
\end{equation*}%
and the full classical capacity is just $C(\mathcal{M}^{\prime })=\max_{%
\overline{\rho }}C(\mathcal{M}^{\prime },\overline{\rho }).$
\end{Rem}


\begin{Th}
\label{th1} \textit{\ The minimization problem }%
\begin{equation}
\min_{\mathcal{E}^{\prime }:\overline{\sigma }=\overline{\rho }}\sum_{k}%
\mathrm{Tr\,}\,\left( p_{k}\sigma _{k}\right) K\left( \sigma _{k}\right)
\label{min}
\end{equation}%
\textit{has the dual problem}
\begin{equation}
\max \left\{ \mathrm{Tr\,}\overline{\rho }\Lambda :\Lambda ^{\ast }=\Lambda
,\,\Lambda \leq K(\sigma )\text{ for all }\sigma \in \mathfrak{S}\right\} .
\label{dprob}
\end{equation}

\textit{The following statements are equivalent:}

\begin{description}
\item[(i)] $\Lambda _{0}$ \textit{is the solution of the problem} (\ref%
{dprob}); ensemble $\mathcal{E}_{0}^{\prime }=\left\{ p_{k}^{0},\sigma
_{k}^{0}\right\} $ with $\sigma _{k}^{0}=\left\vert \phi _{k}^{0}\rangle
\langle \phi _{k}^{0}\right\vert $ \textit{is the solution of the problem} (%
\ref{min}) (equivalently, observable $M_{k}^{0}=\left\vert \varphi
_{k}^{o}\rangle \langle \varphi _{k}^{0}\right\vert ,\,k=1,\dots ,n,$ with%
\begin{equation*}
|\varphi _{k}^{0}\rangle =\left( p_{k}^{0}\right) ^{1/2}\overline{\rho }%
^{-1/2}|\phi _{k}^{0}\rangle ,
\end{equation*}
maximizes the accessible information (\ref{ai}));

\item[(ii)]

\begin{description}
\item[a.] $\Lambda _{0}\leq K(\sigma )$ \textit{for all} $\sigma \in
\mathfrak{S};$

\item[b.] $\left[ K(\sigma _{k}^{0})-\Lambda _{0}\right] \,|\phi
_{k}^{0}\rangle =0,\,k=1,\dots ,n$.
\end{description}

\item[(iii)]

\begin{description}
\item[a.] The entropy inequality%
\begin{equation}
-\sum_{j}\left\langle \psi |M_{j}^{\prime }|\psi \right\rangle \log
\left\langle \psi |M_{j}^{\prime }|\psi \right\rangle \mathrm{\,}\,\geq
\left\langle \psi |\Lambda _{0}|\psi \right\rangle ,  \label{ent}
\end{equation}%
holds for all unit vectors $\psi \in \mathcal{H};$

\item[b.] The unit vectors $|\phi _{k}^{0}\rangle $ turn (\ref{ent}) into
equality:%
\begin{equation}
-\sum_{j}\left\langle \phi _{k}^{0}|M_{j}^{\prime }|\phi
_{k}^{0}\right\rangle \log \left\langle \phi _{k}^{0}|M_{j}^{\prime }|\phi
_{k}^{0}\right\rangle =\left\langle \phi _{k}^{0}|\Lambda _{0}|\phi
_{k}^{0}\right\rangle ,\quad k=1,\dots .n.  \label{enteq}
\end{equation}
\end{description}
\end{description}

The solution of the dual problem (\ref{dprob}) is unique on the support of
the operator $\overline{\rho }$ i.e. $\Lambda _{0}\overline{\rho }=\Lambda
_{0}^{\prime}\overline{\rho }$ for any two solutions $\Lambda _{0},\Lambda
_{0}^{\prime}.$ In particular, if $\overline{\rho }$ is non-degenerate, then
the solution of the dual problem is unique.
\end{Th}

According to (\ref{ae2}), the value of the accessible information (\ref{ae2}%
) is then
\begin{equation}
A(\mathcal{E})=H(\pi )-\mathrm{Tr\,}\overline{\rho }\Lambda _{0}.
\label{acces}
\end{equation}

\begin{proof}
We introduce the POVM on $\mathfrak{S}=\mathfrak{S}(\mathcal{H})$ as%
\begin{equation*}
\Pi (d\sigma )=\sum_{k}\left( p_{k}\sigma _{k}\right) \delta _{\sigma
_{k}}(d\sigma ),
\end{equation*}%
where $\delta _{\sigma _{k}}(d\sigma )$ is the $\delta -$measure
concentrated in the point $\sigma _{k}\in \mathfrak{S.}$ then the problem (%
\ref{min}) can be written as\footnote{%
We refer to \cite{ljm} for rigorous general definition of the traced
integral and formulations of the optimality conditions.}
\begin{eqnarray*}
\mathrm{Tr}\int_{\mathfrak{S}}\,K(\sigma )\Pi (d\sigma ) &\longrightarrow
&\min  
\\
\Pi (d\sigma ) &\geq &0\text{ },  \notag \\
\Pi (\mathfrak{S}) &=&\overline{\rho },  \notag
\end{eqnarray*}%
which is similar to the Bayes estimation problem \cite{jma1}.
This is a convex programming problem for which the dual problem is (\ref%
{dprob}) and the optimality conditions are (ii.a) and%
\begin{equation*}
\int_{B}\left[ K(\sigma )-\Lambda _{0}\right] \Pi _{0}(d\sigma )=0,\quad
B\subset \mathfrak{S}(\mathcal{H}),
\end{equation*}%
see \cite{ljm} for detail. In terms of the optimal ensemble $\mathcal{E}_{0}^{\prime
}=\left\{ p_{k}^{0},\sigma _{k}^{0}\right\} $ the last condition can be
rewritten as%
\begin{equation}
\left[ K(\sigma _{k}^{0})-\Lambda _{0}\right] \,\sigma _{k}^{0}=0,\quad
(p_{k}^{0}>0).  \label{2b}
\end{equation}%
Note that by summing over $k$ we obtain%
\begin{equation}
\sum_{k}K(\sigma _{k}^{0})\,p_{k}^{0}\sigma _{k}^{0}\,=\Lambda _{0}\overline{%
\rho },  \label{3b}
\end{equation}%
where $\Lambda _{0}$ must be Hermitean operator. This implies that $\Lambda _{0}\overline{\rho }$
is the same for all solutions of the dual problem.

By taking into account that $\sigma _{k}^{0}=\left\vert \phi _{k}^{0}\rangle
\langle \phi _{k}^{0}\right\vert ,$ we can rewrite (\ref{2b}) as
\begin{equation}
\left[ K(\sigma _{k}^{0})-\Lambda _{0}\right] |\phi _{k}^{0}\rangle =0
\label{ii2}
\end{equation}%
and (\ref{3b}) as
\begin{equation*}
\sum_{k}K(\sigma _{k}^{0})\,p_{k}^{0}\left\vert \phi _{k}^{0}\rangle \langle
\phi _{k}^{0}\right\vert =\Lambda _{0}\overline{\rho }.  \label{l0}
\end{equation*}

Taking into account (\ref{K}), the condition (ii.a) becomes%
\begin{equation*}
-\sum_{j}\left\langle \psi |M_{j}^{\prime }|\psi \right\rangle \log \mathrm{%
Tr\,}\,\sigma M_{j}^{\prime }\geq \left\langle \psi |\Lambda _{0}|\psi
\right\rangle ,\quad \psi \in \mathcal{H},
\end{equation*}%
for all $\sigma \in \mathfrak{S.}$ Without loss of generality we can assume
that $\left\Vert \psi \right\Vert =1$ here. Then for fixed $\psi $ the
minimal value of the lefthand side is
\begin{equation*}
-\sum_{j}\left\langle \psi |M_{j}^{\prime }|\psi \right\rangle \log
\left\langle \psi |M_{j}^{\prime }|\psi \right\rangle
\end{equation*}%
since the difference is just the relative entropy of two probability
distributions.

Thus we arrive to the entropy inequality (\ref{ent}) as an equivalent of the
condition (ii.a). Further, by taking into account that $K(\sigma
_{k}^{0})-\Lambda _{0}\geq 0,$ the condition (ii.b) in the form (\ref{ii2})
is the same as
\begin{equation*}
\langle \phi _{k}^{0}|K(\sigma _{k}^{0})|\phi _{k}^{0}\rangle =\langle \phi
_{k}^{0}|\Lambda _{0}|\phi _{k}^{0}\rangle
\end{equation*}%
which via (\ref{K}) means that the unit vectors $|\phi _{k}^{0}\rangle $,
satisfying the constraint $\sum_{k}p_{k}^{0}\left\vert \phi _{k}^{0}\rangle
\langle \phi _{k}^{0}\right\vert =\overline{\rho },$ turn (\ref{ent}) into
the equality (\ref{enteq}).
\end{proof}

Especially important for us is the case of pure-states ensemble with $\rho
_{j}=|\psi _{j}\rangle \langle \psi _{j}|$ where $|\psi _{j}\rangle
;j=1,\dots ,d$, are linearly independent vectors. In this case the
resolution of the identity (\ref{mprim}) is
\begin{equation*}
M_{j}^{\prime }=|e_{j}\rangle \langle e_{j}|,
\end{equation*}%
where the vectors
\begin{equation*}
|e_{j}\rangle =\sqrt{\pi _{j}}\bar{\rho}^{-1/2}|\psi _{j}\rangle ;\quad
j=1,\dots ,d,
\end{equation*}%
by necessity form an orthonormal basis (depending on the ensemble $\mathcal{E%
}$). Then the entropy inequality (\ref{ent}) takes the form
\begin{equation*}
-\sum_{j=1}^{d}\left\vert z_{j}\right\vert ^{2}\log \left\vert
z_{j}\right\vert ^{2}\geq \sum_{j=1}^{d}\lambda _{jk}\overline{z_{j}}z_{k},
\label{ent1}
\end{equation*}%
where $\lambda _{jk}=\langle e_{j}|\Lambda _{0}|e_{k}\rangle $ and $%
z_{j}=\langle e_{j}|\psi \rangle $ are complex variables subject only to the
constraint \textit{\ }%
\begin{equation}
\sum_{j=1}^{m}\left\vert z_{j}\right\vert ^{2}=1.  \label{condw1}
\end{equation}
The left-hand side is just the Shannon entropy of the p.d. $\left\{
\left\vert z_{j}\right\vert ^{2};j=1,\dots ,d\right\} $, and such an
inequality can be approached with the tools of classical analysis. We will
need the following

\begin{Lem}
\label{lemma} 
Let $\lambda _{jk}$ be real numbers. If the inequality
\begin{equation}
\sum_{j=1}^{m}\left\vert z_{j}\right\vert ^{2}\log \left\vert
z_{j}\right\vert ^{2}+\sum_{j=1}^{d}\lambda _{jk}\overline{z_{j}}z_{k}\leq 0
\label{ineql}
\end{equation}%
holds for all real $z_{j}$ satisfying the condition (\ref{condw1}), then it
holds for all complex $z_{j}$ satisfying the same condition.
\end{Lem}

\begin{proof}
Let $z_{j}=x_{j}+iy_{j}$ for real $x_{j},y_{j}$, then (\ref{condw1}) amounts
to
\begin{equation*}
\sum_{j=1}^{m}x_{j}^{2}+\sum_{j=1}^{m}y_{j}^{2}=1.
\end{equation*}

Denote $P=\sum_{j=1}^{m}x_{j}^{2},\,1-P=\sum_{j=1}^{m}y_{j}^{2}$. The
function $f(\xi )=\xi \log \xi $ is convex, and $\lambda _{jk}$ are real
hence
\begin{eqnarray*}
\sum_{j=1}^{m}\left\vert z_{j}\right\vert ^{2}\log \left\vert
z_{j}\right\vert ^{2}+\sum_{j=1}^{d}\lambda _{jk}\overline{z_{j}}z_{k}
&=&\sum_{j=1}^{m}f\left( P\left( \frac{x_{j}}{\sqrt{P}}\right)
^{2}+(1-P)\left( \frac{y_{j}}{\sqrt{1-P}}\right) ^{2}\right)  \\
&+&P\sum_{j=1}^{d}\lambda _{jk}x_{j}x_{k}+(1-P)\sum_{j=1}^{d}\lambda
_{jk}y_{j}y_{k} \\
&\leq &P\left[ \sum_{j=1}^{m}f\left( \left( \frac{x_{j}}{\sqrt{P}}\right)
^{2}\right) +\sum_{j=1}^{d}\lambda _{jk}x_{j}x_{k}\right]  \\
&+&(1-P)\left[ \sum_{j=1}^{m}f\left( \left( \frac{y_{j}}{\sqrt{1-P}}\right)
^{2}\right) +\sum_{j=1}^{d}\lambda _{jk}y_{j}y_{k}\right] \leq 0,
\end{eqnarray*}%
because
\begin{equation*}
\sum_{j=1}^{m}\left( \frac{x_{j}}{\sqrt{P}}\right) ^{2}=\sum_{j=1}^{m}\left(
\frac{y_{j}}{\sqrt{1-P}}\right) ^{2}=1,\quad
\end{equation*}%
and by assumption (\ref{ineql}) holds for real $z_{j}$ satisfying (\ref%
{condw1}).  \end{proof}

The optimality conditions can be useful in the following ways:

\begin{enumerate}
\item for a problem which has a hypothetical solution, the conditions can be
used to verify its actual optimality. One inserts the conjectured $%
p_{k}^{0},\sigma _{k}^{0}$ into (\ref{2b}) to compute Hermitean $\Lambda
_{0}\ $. Then the main difficulty will be proof of the entropy inequality (%
\ref{ent}). In the continuous variable version, this was our strategy in the
solution of the related Gaussian optimizers conjecture for the capacity of
Gaussian measurement channels \cite{ljm1}, with the new generalizations of
the famous log-Sobolev inequality as the entropy inequality.

\item for a problem which has a known solution, we get an entropy inequality
(\ref{ent}) as a consequence of the optimality. Also one can obtain an
alternative proof of the optimality if one can find an independent proof of
the entropy inequality. In the next section we illustrate this point with a
simple example resulting in the yet nontrivial tight lower bounds for the
binary Shannon entropy.

\item The uniqueness of the solution of the dual problem (assuming $\bar{\rho%
}$ non-degenerate) has the following useful consequence. Assume that the
ensemble $\mathcal{E}=\left\{ \pi _{j},\rho _{j}\right\} _{j=1,\dots ,m}$ is
\textit{symmetric} in that $\pi _{j}\equiv 1/m$ and there is a group of
*-automorphisms $\left\{ \alpha _{g}; g\in G\right\} $ of the algebra of
operators on $\mathcal{H}$ acting transitively on the set of states $\left\{
\rho _{j}\right\} _{j=1,\dots ,m}$. Then $\Lambda _{0}$ is an invariant of
this group: $\alpha _{g}^{\ast }\left[ \Lambda _{0}\right] =\Lambda
_{0};g\in G.$
\end{enumerate}

\section{Example: two pure states}

\label{sec:two}

The case of ensemble of two states was treated previously in \cite{lev},
\cite{keil}. For simplicity of demonstration we take two pure equiprobable
states. Let $|e_{0}\rangle ,|e_{1}\rangle $ be an orthonormal basis in
two-dimensional Hilbert space. Consider ensemble $\mathcal{E}=\left\{ \pi
_{\pm },\rho _{\pm }\right\} ,$ where $\pi _{\pm }=\frac{1}{2},\rho _{\pm
}=\left\vert \psi _{\pm }\rangle \langle \psi _{\pm }\right\vert $ with $%
|\psi _{\pm }\rangle =\cos \alpha /2|e_{0}\rangle \pm \sin \alpha
/2|e_{1}\rangle ,\,0\leq \alpha \leq \pi .$ As argued in \cite{lev} the
accesible information%
\begin{equation}
A(\mathcal{E})=1-h\left( \frac{1+\sin \alpha }{2}\right)  \label{opt}
\end{equation}%
is attained by the optimal observable $\mathcal{M=}\left\{
M_{+},M_{-}\right\} $ where $M_{\pm }=\left\vert e_{\pm }\rangle \langle
e_{\pm }\right\vert $ and $|e_{\pm }\rangle =\frac{1}{\sqrt{2}}\left(
|e_{0}\rangle \pm |e_{1}\rangle \right) $ is an orthonormal basis\footnote{%
Strictly speaking, the paper \cite{lev} contains the check of \textit{local}
optimality.}. Here
\begin{equation*}
h(t)=-t\log t-(1-t)\log (1-t)
\end{equation*}%
is the classical binary entropy.

Let us apply theorem \ref{th1} giving a proof of the global optimality and,
as a byproduct, a new entropy inequality. The average state of the ensemble $%
\mathcal{E}$ is
\begin{equation*}
\overline{\rho }=\frac{1}{2}\left\vert \psi _{+}\rangle \langle \psi
_{+}\right\vert +\frac{1}{2}\left\vert \psi _{-}\rangle \langle \psi
_{-}\right\vert =\cos ^{2}\alpha /2\left\vert e_{0}\rangle \langle
e_{0}\right\vert +\sin ^{2}\alpha /2\left\vert e_{1}\rangle \langle
e_{1}\right\vert
\end{equation*}%
with%
\begin{equation*}
\overline{\rho }^{1/2}=\cos \alpha /2\left\vert e_{0}\rangle \langle
e_{0}\right\vert +\sin \alpha /2\left\vert e_{1}\rangle \langle
e_{1}\right\vert  \label{1/2}
\end{equation*}%
so that
\begin{equation*}
\frac{1}{\sqrt{2}}\overline{\rho }^{-1/2}|\psi _{\pm }\rangle =|e_{\pm
}\rangle  \label{eq1}
\end{equation*}%
and
\begin{equation}
M_{\pm }^{\prime }=\frac{1}{2}\overline{\rho }^{-1/2}\rho _{\pm }\overline{%
\rho }^{-1/2}=\left\vert e_{\pm }\rangle \langle e_{\pm }\right\vert ,
\label{mprime}
\end{equation}%
which is also the conjectured optimal measurement. Therefore the conjectured
optimal ensemble $\mathcal{E}_{0}^{\prime }$ coincides with $\mathcal{E}$ \
because
\begin{equation*}
\sigma _{\pm }^{0}=2\overline{\rho }^{1/2}\left\vert e_{\pm }\rangle \langle
e_{\pm }\right\vert \overline{\rho }^{1/2}=\left\vert \psi _{\pm }\rangle
\langle \psi _{\pm }\right\vert ,\quad \pi _{\pm }^{0}=\frac{1}{2},
\end{equation*}%
$\,$ so that $|\phi _{\pm }^{0}\rangle =|\psi _{\pm }\rangle .$ Let us check
the optimality conditions for the ensemble $\mathcal{E}^{0}=\mathcal{E}$.

Introducing the important parameter
\begin{equation*}
p=\frac{1+\sin \alpha }{2},\quad \frac{1}{2}\leq p\leq 1,  \label{par}
\end{equation*}%
we have
\begin{equation*}
\left\langle e_{\pm }|\psi _{\pm }\right\rangle =\sqrt{p},\quad \left\langle
e_{\pm }|\psi _{\mp }\right\rangle =\sqrt{1-p}.
\end{equation*}%
Incerting $\sigma _{\pm }^{0}=\left\vert \psi _{\pm }\rangle \langle \psi
_{\pm }\right\vert $ into%
\begin{equation*}
K(\sigma )=-\sum_{\pm }\left\vert e_{\pm }\rangle \langle e_{\pm
}\right\vert \log \left\langle e_{\pm }|\sigma |e_{\pm }\right\rangle
\end{equation*}%
we obtain
\begin{eqnarray*}
K(\sigma _{+}^{0}) &=&-\left\vert e_{+}\rangle \langle e_{+}\right\vert \log
p-\left\vert e_{-}\rangle \langle e_{-}\right\vert \log (1-p), \\
K(\sigma _{-}^{0}) &=&-\left\vert e_{+}\rangle \langle e_{+}\right\vert \log
(1-p)-\left\vert e_{-}\rangle \langle e_{-}\right\vert \log p.
\end{eqnarray*}

The ensemble $\mathcal{E}$ is symmetric with respect to the reflection
relative to the axis $|e_{0}\rangle ;$ according to the remark 3 at the end
of previous section, $\Lambda _{0}$ should also obey this symmetry, hence%
\begin{equation}
\Lambda _{0}=\lambda _{0}(p)\left\vert e_{0}\rangle \langle e_{0}\right\vert
-\lambda _{1}(p)\left\vert e_{1}\rangle \langle e_{1}\right\vert .
\label{lam0}
\end{equation}%
Solving the conditions (ii.b)%
\begin{equation*}
\left[ K(\sigma _{\pm }^{0})-\Lambda _{0}\right] \,|\phi _{\pm }^{0}\rangle
=0
\end{equation*}%
and taking into account that
\begin{equation*}
\cos \alpha /2=\frac{1}{\sqrt{2}}\left( \sqrt{p}+\sqrt{1-p}\right) ,\quad
\sin \alpha /2=\frac{1}{\sqrt{2}}\left( \sqrt{p}-\sqrt{1-p}\right) ,
\end{equation*}%
results in%
\begin{equation*}
\lambda _{0}(p)=-\frac{\sqrt{p}\log p+\sqrt{1-p}\log \left( 1-p\right) }{%
\sqrt{p}+\sqrt{1-p}},
\end{equation*}%
\begin{equation*}
\lambda _{1}(p)=\frac{\sqrt{p}\log p-\sqrt{1-p}\log \left( 1-p\right) }{%
\sqrt{p}-\sqrt{1-p}},
\end{equation*}%
provided $p>1/2$.

Then the conditions (ii.b) are fulfilled while (ii.a) reduces to (\ref{ent})
which via (\ref{mprime}), (\ref{lam0}) amounts to the following parametric
family of tight lower bounds for the classical binary entropy

\begin{prop}
For $p>1/2$ and $0\leq t\leq 1$ (in our case $t=|\left\langle e_{+}|\psi
\right\rangle |^{2},1-t=|\left\langle e_{-}|\psi \right\rangle |^{2}$),
\begin{equation}  \label{ineq}
h(t)\geq (\lambda _{0}(p)+\lambda _{1}(p))\sqrt{t(1-t)}-\frac{\lambda
_{1}(p)-\lambda _{0}(p)}{2},
\end{equation}%
where
\begin{equation*}
\frac{\lambda _{0}(p)+\lambda _{1}(p)}{2}=\sqrt{p\left( 1-p\right) }\frac{%
\log p-\log \left( 1-p\right) }{2p-1}\equiv \mu _{0}(p),
\end{equation*}%
\begin{equation*}
\frac{\lambda _{1}(p)-\lambda _{0}(p)}{2}=\frac{p\log p-\left( 1-p\right)
\log \left( 1-p\right) }{2p-1}.
\end{equation*}%
The inequality (\ref{ineq}) becomes equality for $t=p$ and $t=1-p.$
\end{prop}

In these relations $p>1/2;$ taking the limit $p\rightarrow 1/2$ the
inequality (\ref{ineq}) becomes%
\begin{equation}
h(t)\geq 2\log e\sqrt{t(1-t)}-\log \frac{e}{2}.  \label{ineq0}
\end{equation}%
The value $p=1/2$ corresponds to the trivial case $\rho _{+}=\rho _{-}$ with
$A(\mathcal{E})=0$. It is remarkable, however, that in the limit $%
p\rightarrow 1/2$ we get the nontrivial entropy inequality (\ref{ineq0}). It
is equivalent to an inequality which was a starting point in the original
derivation of logarithmic Sobolev inequality by L. Gross \cite{gross}.

Thus (\ref{ineq}), (\ref{ineq0}) can be regarded as discrete
\textquotedblleft distant relatives\textquotedblright\ of the log-Sobolev
inequality. In the paper \cite{ljm1} we elaborated continuous variable
version of the optimality conditions and applied it in the bosonic Gaussian
case; the corresponding entropy inequality turned out to be a generalization
of the genuine log-Sobolev inequality.

\begin{proof} 
For convenience of differentiation we pass to natural logarithms. Denote by $%
g(t,p)$ the righthand side of (\ref{ineq}), then
\begin{eqnarray}
\frac{d}{dt}h(t) &=&-\ln t+\ln (1-t),  \label{der1} \\
\frac{d}{dt}g(t,p) &=&\mu _{0}(p)\frac{1-2t}{\sqrt{t\left( 1-t\right) }}.
\label{der2}
\end{eqnarray}%
Taking into account the symmetry of the functions with respect to the point $%
t=1/2,$ we can restrict to the interval $t\in \lbrack 1/2,1].$ One can check
that $h(p)=g(p,p)$ and $\frac{d}{dt}|_{t=p}h(t)=\frac{d}{dt}|_{t=p}g(t,p)$
which means that $h(t)$ is the envelop of the family $g(t,p);\,p\geq 1/2,$
with $t=p$ as the tangential points. By using (\ref{der1}), (\ref{der2}) we
obtain the key relation
\begin{equation*}
\frac{d}{dt}\left[ h(t)-g(t,p)\right] =\left[ \mu _{0}(p)-\mu _{0}(t)\right]
\frac{2t-1}{\sqrt{t\left( 1-t\right) }}.
\end{equation*}

The function
\begin{equation*}
\mu _{0}(t)=\sqrt{t\left( 1-t\right) }\frac{\ln t-\ln \left( 1-t\right) }{%
2t-1}
\end{equation*}
is monotonously decreasing on $[1/2,1],$ which can be seen by showing its
derivative is negative on $(1/2,1]$ (for this it is convenient to make a
monotonic change of variable $u=\sqrt{t/\left( 1-t\right) }$).

Thus $h(t)-g(t,p)=0$ for $t=p,$ and $\frac{d}{dt}\left[ h(t)-g(t,p)\right] $
has the same sign as $t-p.$ Hence $h(t)-g(t,p)\geq 0.$
\end{proof}

Theorem \ref{th1} then implies the global optimality of the observable $%
\mathcal{M=}\left\{ M_{+},M_{-}\right\} .$ Also
\begin{equation*}
\mathrm{Tr\,}\overline{\rho }\Lambda _{0}=\lambda _{0}(p)\cos ^{2}\alpha
/2-\lambda _{1}(p)\sin ^{2}\alpha /2=h(p)
\end{equation*}%
and (\ref{acces}) implies $A(\mathcal{E})=1-h(p)$ in accordance with (\ref%
{opt}).

\section{Multidimensional case}

\label{sec:multi}

\subsection{The lower bound for the Shannon entropy: formulation and
discussion}

\label{conj}

The multidimensional generalization of (\ref{ineq}) to be derived in what
follows is:

\begin{Th}
\label{th2} \textit{For integer $m\geq 2$ and} $p>\frac{m-1}{m}$
\begin{equation}
-\sum_{j=1}^{m}t_{j}\log t_{j}\geq \mu _{0}(p)\left( \sum_{j=1}^{m}\sqrt{%
t_{j}}\right) ^{2}-\mu _{1}(p),  \label{basic}
\end{equation}%
\textit{where $t_{j}\geq 0,\sum_{j=1}^{m}t_{j}=1$ and}
\begin{equation}
\mu _{0}(p)=\frac{\sqrt{\frac{p\left( 1-p\right) }{m-1}}\left( \log p-\log
\frac{1-p}{m-1}\right) }{\left( \sqrt{p}+\left( m-1\right) \sqrt{\frac{1-p}{%
m-1}}\right) \left( \sqrt{p}-\sqrt{\frac{1-p}{m-1}}\right) },  \label{mu0}
\end{equation}%
\begin{equation}
\mu _{1}(p)=\frac{\sqrt{p}\log p-\sqrt{\frac{1-p}{m-1}}\log \frac{1-p}{m-1}}{%
\left( \sqrt{p}-\sqrt{\frac{1-p}{m-1}}\right) }.  \label{mu1}
\end{equation}

This lower bound is \textit{tight}: the graphs of the functions of $%
t_{1},\dots ,t_{m} $ in the left and in the right sides of (\ref{basic}) are
tangential at the points obtained by permutations of $t_{1}=p,t_{k}=\frac{1-p%
}{m-1},k=2,\dots ,m$.
\end{Th}

The inequality (\ref{basic}) arises in connection with the problem of
accessible information for ensemble of $m$ equiprobable equiangular pure
states discussed in the next subsection. The suggestion for a proof of (\ref%
{basic}) will be given in section \ref{prroof}. Note that in the case $m=2$
the inequality (\ref{basic}) reduces to (\ref{ineq}).

For $m\geq 3$ the value $p=\frac{m-1}{m}$ is allowed giving the inequality
\begin{eqnarray}
-\sum_{j=1}^{m}t_{j}\log t_{j} &\geq &\frac{\log \left( m-1\right)}{m-2}
\left( \sum_{j=1}^{m}\sqrt{t_{j}}\right) ^{2}-\left[ \frac{m\log \left(
m-1\right)}{m-2} -\log m\right]\quad  \notag
\\
&=&\log m-\frac{m\log \left( m-1\right) }{m-2}\left[ 1-B(P;P_{U})^{2}\right]
,  \label{basic11}
\end{eqnarray}%
where%
\begin{equation*}
B(P;P_{U})=\sum_{j=1}^{m}\sqrt{t_{j}/m}
\end{equation*}
is the Bhattacharyya coefficient between the probability distribution $%
P=\left( t_{1},\dots ,t_{m}\right) $ and the uniform distribution $%
P_{U}=\left( 1/m,\dots ,1/m\right) .$

The inequality (\ref{basic11}) then gives a lower bound for the Shannon
entropy in terms of a \textquotedblleft distance\textquotedblright\ between
the distribution $P$ and the uniform distribution $P_{U}$, related to the
Bhattacharyya coefficient. It can be compared with the lower bound of
theorem 3.11 in \cite{moser}:%
\begin{equation}
-\sum_{j=1}^{m}t_{j}\log t_{j}\geq \log m-\frac{m\log m}{2\left( m-1\right) }%
D(P;P_{U}),  \label{mos}
\end{equation}%
where $D(P;P_{U})=\sum_{j=1}^{m}\left\vert t_{j}-1/m\right\vert $ is the
total variation distance. There is well-known relation for any two
distributions $P, P^{\prime }$:
\begin{equation*}
2\left( 1-B(P;P^{\prime })\right) \leq D(P;P^{\prime })\leq 2\sqrt{%
1-B(P;P^{\prime })^{2}},
\end{equation*}%
which allows to obtain from (\ref{basic11}) a degraded version of (\ref{mos}%
) and vice versa. For $P$ close to $P_{U},$ the bound (\ref{basic11}) is
substantially better than (\ref{mos}) while close to degenerate
distributions with $B(P;P_{U})\approx \sqrt{1/m}$ the right-hand side of (%
\ref{basic11}) becomes negative which never happens for (\ref{mos}).

Similarly to (\ref{mos}), the inequality (\ref{basic11}) (as well as (\ref%
{basic})) can be rewritten as a tight upper bound for the classical relative
entropy (the Kullback-Leibler divergence)
\begin{equation}\label{hacute}
H(P||P_{U})\leq \frac{m\log \left( m-1\right) }{m-2}\left[ 1-B(P;P_{U})^{2}%
\right].
\end{equation}

\subsection{The conjecture: \textquotedblleft quantum
pyramids\textquotedblright}

\label{pyram}

Let $|e_{j}\rangle ,\,j=1,\dots ,m,$ be an orthonormal basis in $m-$%
dimensional Hilbert space.
\begin{equation*}
|e_{0}\rangle =\frac{1}{\sqrt{m}}\sum_{j=1}^{m}|e_{j}\rangle ,\quad \langle
e_{0}|e_{j}\rangle \equiv \frac{1}{\sqrt{m}}.
\end{equation*}%
Following Englert and \v{R}eh\'{a}\v{c}ek \ \cite{eng} we consider the
system of equiangular vectors -- the \textit{quantum pyramid}:
\begin{equation}
|\psi _{j}\rangle =\sqrt{mr_{1}}|e_{j}\rangle +\left( \sqrt{r_{0}}-\sqrt{%
r_{1}}\right) |e_{0}\rangle ,\quad j=1,\dots ,m,  \label{pyr}
\end{equation}%
\begin{equation}
r_{0},\,r_{1}\geq 0,\quad r_{0}+(m-1)r_{1}=1.  \label{r01}
\end{equation}%
Alternative parametrization was given in \cite{hel}, section VI.2, where
this configuration was related to \textquotedblleft orthogonal
signals\textquotedblright\ in communication theory:%
\begin{equation}
|\psi _{j}\rangle =a|e_{j}\rangle +b\sum_{k\neq j}|e_{k}\rangle ,\quad
j=1,\dots ,m,  \label{pyr1}
\end{equation}%
\begin{equation*}
a^{2}+(m-1)b^{2}=1,\quad a\geq \frac{1}{\sqrt{m}}.  \label{ab}
\end{equation*}%
The relation between the two parametrizations is%
\begin{equation}
a=\sqrt{mr_{1}}+\frac{\sqrt{r_{0}}-\sqrt{r_{1}}}{\sqrt{m}},\quad b=\frac{%
\sqrt{r_{0}}-\sqrt{r_{1}}}{\sqrt{m}},  \label{param}
\end{equation}%
\begin{equation}
1-r_{0}=(m-1)r_{1}=\frac{m-1}{m}(a-b)^2.  \label{param1}
\end{equation}

The cosine of angle between any two different edges is
\begin{equation}
\xi \equiv \langle \psi _{j}|\psi _{k}\rangle
=r_{0}-r_{1}=2ab+(m-2)b^{2}\,,\quad \,j\neq k,  \label{cosine}
\end{equation}%
thus the case $r_{0}>r_{1}\,(b>0)$ corresponds to the \textit{acute }($0<\xi
<1$), $r_{0}=r_{1}\,(b=0)$ -- to the \textit{orthogonal }($\xi =0$), and $%
r_{0}<r_{1}\,(b<0)$ -- to the \textit{obtuse} $\left( -\frac{1}{\left(
m-1\right) }<\xi <0\right) $ pyramids. The extreme cases are $%
r_{1}=0\,\left( a=b=\sqrt{\frac{1}{m}}\right) $ -- the \textit{ort} $|\psi
_{j}\rangle \equiv |e_{0}\rangle $ ($\xi =1$)$,$ and $r_{0}=0\,\left( a=%
\sqrt{\frac{m-1}{m}},b=-\sqrt{\frac{1}{m\left( m-1\right) }}\right) $ -- the
\textit{flat} $\left( \xi =-\frac{1}{\left( m-1\right) }\right) $ pyramid.

We will consider the ensemble $\mathcal{E}=\left\{ \pi _{j},\rho
_{j}\right\} _{j=1,\dots ,m}$ with $\pi _{j}=\frac{1}{m},\,\rho _{j}=|\psi
_{j}\rangle \langle \psi _{j}|$ with the average state%
\begin{equation}
\bar{\rho}\equiv \frac{1}{m}\sum_{j=1}^{m}|\psi _{j}\rangle \langle \psi
_{j}|=r_{1}I+\left( r_{0}-r_{1}\right) |e_{0}\rangle \langle e_{0}|.
\label{rhobar}
\end{equation}

\begin{Rem}
\label{obs} Operators of the form $c_{1}I+c_{0}|e_{0}\rangle \langle e_{0}|$
constitute Abelian algebra of operators invariant under the group of all
*-automorphisms generated by the orthogonal transformations leaving the
quantum pyramid invariant.
\end{Rem}

One can then easily check that%
\begin{equation}  \label{sqrt}
\bar{\rho}^{1/2}=\sqrt{r_{1}}I+\left( \sqrt{r_{0}}-\sqrt{r_{1}}\right)
|e_{0}\rangle \langle e_{0}|
\end{equation}%
and%
\begin{equation}
|\psi _{j}\rangle =\sqrt{m}\bar{\rho}^{1/2}|e_{j}\rangle .  \label{psie}
\end{equation}

In what follows we consider the solution for the accessible information $A(%
\mathcal{E})$ and its optimizing observable conjectured and supported
numerically in \cite{eng}. In the rest of this subsection and in the next
subsection we consider acute pyramids for which $r_{0}\geq r_{1}$ or,
equivalently, $b\geq 0.$ 
According to \cite{eng} the conjectured optimal observable belongs to the
family depending on the parameter $0\leq t\leq 1$ and has the form $\mathcal{%
M}=\left\{ M_{k};k=0,1,\dots ,m\right\} $ where $M_{k}=|\tilde{e}_{k}\rangle
\langle \tilde{e}_{k}|$%
\begin{equation}
|\tilde{e}_{k}\rangle =|e_{k}\rangle +\frac{t-1}{\sqrt{m}}|e_{0}\rangle
;\quad k=1,\dots ,m,  \label{etilde}
\end{equation}%
\begin{equation*}
|\tilde{e}_{0}\rangle =\sqrt{1-t^{2}}|e_{0}\rangle .
\end{equation*}%
The optimal value of the parameter is%
\begin{equation}
t_{a}(m)=\left\{
\begin{array}{cc}
1, & r_{0}\leq \frac{4(m-1)}{m^{2}} \\
\frac{2(m-1)}{m-2}\sqrt{\frac{r_{1}}{r_{0}}}, & r_{0}>\frac{4(m-1)}{m^{2}}%
\end{array}%
\right. ,  \label{t0}
\end{equation}%
giving the accessible information (eq. (36) in \cite{eng})%
\begin{equation}
A(\mathcal{E})=\left\{
\begin{array}{cc}
I^{(SRM)}, & r_{0}\leq \frac{4(m-1)}{m^2} \\
\frac{m\left( 1-r_{0}\right) }{m-2}\log (m-1), & r_{0}>\frac{4(m-1)}{m^2}%
\end{array}%
\right. ,  \label{acc}
\end{equation}%
where%
\begin{eqnarray}
I^{(SRM)} &=&\frac{1}{m}\left( \sqrt{r_{0}}+\left( m-1\right) \sqrt{r_{1}}%
\right) ^{2}\log \left( \sqrt{r_{0}}+\left( m-1\right) \sqrt{r_{1}}\right)
^{2}  \notag \\
&+&\frac{m-1}{m}\left( \sqrt{r_{0}}-\sqrt{r_{1}}\right) ^{2}\log \left(
\sqrt{r_{0}}-\sqrt{r_{1}}\right) ^{2} \notag 
\\
&=&\log m+a^{2}\log a^{2}+\left( m-1\right) b^{2}\log b^{2}.  \label{isrm1}
\end{eqnarray}%
Note that in the case $t_{a}(m)=1$ the conjectured information-optimal
observable $M_{k}=|{e}_{k}\rangle \langle {e}_{k}|,\,k=1,\dots ,m,$ is
optimal for the Bayes discrimination problem for the ensemble $\mathcal{E}$
\cite{hel}, sec. VI.2. It is called square-root measurement (SRM) since the
system $\left\{ |{e}_{k}\rangle \right\} $ is obtained from $\left\{ |{\psi }%
_{k}\rangle \right\} $ by applying inverse square root of its Gram matrix,
cf. (\ref{psie}).

\begin{Rem}
\label{chi} An upper bound and a benchmark for $A(\mathcal{E})$ is the
classical capacity of the c-q channel $j\rightarrow |\psi _{j}\rangle
\langle \psi_{j}|$:
\begin{equation*}  \label{cap}
C=\max_{\pi}H\left(\sum_{j=1}^{m}\pi_j |\psi _{j}\rangle \langle
\psi_{j}|\right)= H(\bar{\rho}).
\end{equation*}
The last equality follows from the symmetry of the quantum pyramid and
concavity of the von Neumann entropy. According to (\ref{rhobar}) the
density operator $\bar{\rho}$ has the eigenvalues $r_{0}, r_{1}$, the last
of multiplicity $m-1$, hence
\begin{eqnarray*}
C &=&-r_{0}\log r_{0}-(m-1) r_{1}\log r_{1} \\
&=& - \left(1-\frac{m-1}{m}(a-b)^2\right)\log\left(1-\frac{m-1}{m}%
(a-b)^2\right) \\
&-&\frac{m-1}{m}(a-b)^2\log\frac{(a-b)^2}{m},
\end{eqnarray*}
where in the second equality we used (\ref{param1}).

Unlike the case of accessible information, this expression is the same for
all values of the parameters, for acute, obtuse and flat pyramids. The ratio
$C/A(\mathcal{E})$ gives the gain due to the measurement entanglement at the
channel's output. In particular, it tends to the infinity as the pyramid
collapses to ort (weak signal) and equals to 1 for rectangular pyramid
(ideal channel).
\end{Rem}

\subsection{Checking the optimality conditions: the case of acute pyramids}

\label{acute}

Let us apply our approach of section \ref{aopt}. By using (\ref{psie}) the
dual observable $\mathcal{M}^{\prime }$ has the components
\begin{equation}
M_{j}^{\prime }=\overline{\rho }^{-1/2}\left( \frac{1}{m}|\psi _{j}\rangle
\langle \psi _{j}|\right) \overline{\rho }^{-1/2}=|e_{j}\rangle \langle
e_{j}|,  \label{dualobs}
\end{equation}%
hence%
\begin{equation*}
K\left( \sigma \right) =-\sum_{j}|e_{j}\rangle \langle e_{j}|\,\log \mathrm{%
\,}\,\langle e_{j}|\sigma |e_{j}\rangle .
\end{equation*}%
From (\ref{etilde}), (\ref{sqrt}) we obtain the \textit{dual} ensemble $%
\mathcal{E}^{\prime }=\left\{ p_{k},\sigma _{k}\right\} $ with
\begin{equation*}
\sigma _{k}=p_{k}^{-1}|\tilde{\phi}_{k}\rangle \langle \tilde{\phi}%
_{k}|\,,\quad p_{k}=\langle \tilde{\phi}_{k}|\tilde{\phi}_{k}\rangle
\end{equation*}%
where
\begin{eqnarray*}
|\tilde{\phi}_{k}\rangle &=&\overline{\rho }^{1/2}|\tilde{e}_{k}\rangle =%
\sqrt{r_{1}}|e_{k}\rangle +\frac{\sqrt{r_{0}}t-\sqrt{r_{1}}}{\sqrt{m}}%
|e_{0}\rangle ; \\
p_{k} &=&r_{1}+\frac{r_{0}t^{2}-r_{1}}{m};\quad k=1,\dots ,m;
\end{eqnarray*}%
\begin{equation*}
|\tilde{\phi}_{0}\rangle =\overline{\rho }^{1/2}|\tilde{e}_{0}\rangle =\sqrt{%
r_{0}\left( 1-t^{2}\right) }|e_{0}\rangle ;\quad p_{0}=r_{0}\left(
1-t^{2}\right) .
\end{equation*}

\bigskip Substitution gives%
\begin{equation*}
K\left( \sigma _{k}\right) =-\sum_{j}|e_{j}\rangle \langle e_{j}|\,\log
p_{k}^{-1}\mathrm{\,}\,\left\vert \langle e_{j}|\tilde{\phi}_{k}\rangle
\right\vert ^{2},
\end{equation*}%
hence%
\begin{equation*}
K\left( \sigma _{k}\right) =\left\{
\begin{array}{cc}
-\sum_{j}|e_{j}\rangle \langle e_{j}|\,\log p_{k}^{-1}\mathrm{\,}%
\,\left\vert \sqrt{r_{1}}\delta _{jk}+\frac{\sqrt{r_{0}}t-\sqrt{r_{1}}}{m}%
\right\vert ^{2}, & k\geq 1 \\
-\sum_{j}|e_{j}\rangle \langle e_{j}|\,\log \frac{1}{m}=\left( \log m\right)
I, & k=0%
\end{array}%
\right. .
\end{equation*}%
Having in mind checking the optimality condition (ii.b) we compute%
\begin{eqnarray*}
K\left( \sigma _{k}\right) |\tilde{\phi}_{k}\rangle
&=&-\sum_{j}|e_{j}\rangle \left( \sqrt{r_{1}}\delta _{jk}+\frac{\sqrt{r_{0}}%
t-\sqrt{r_{1}}}{m}\right) \,\log \,\frac{\left\vert \sqrt{r_{1}}\delta _{jk}+%
\frac{\sqrt{r_{0}}t-\sqrt{r_{1}}}{m}\right\vert ^{2}}{r_{1}+\frac{%
r_{0}t^{2}-r_{1}}{m}}  \notag \\
&=&A(t)|e_{k}\rangle +B(t)\sqrt{m}|e_{0}\rangle ;\quad k=1,\dots ,m,
\label{AB}
\end{eqnarray*}%
where%
\begin{eqnarray}
A(t) &=&-\left( \sqrt{r_{1}}+\frac{\sqrt{r_{0}}t-\sqrt{r_{1}}}{m}\right)
\,\log \,\frac{\left\vert \sqrt{r_{1}}+\frac{\sqrt{r_{0}}t-\sqrt{r_{1}}}{m}%
\right\vert ^{2}}{r_{1}+\frac{r_{0}t^{2}-r_{1}}{m}}-B(t),  \label{AB1} \\
B(t) &=&-\left( \frac{\sqrt{r_{0}}t-\sqrt{r_{1}}}{m}\right) \,\log \,\frac{%
\left\vert \frac{\sqrt{r_{0}}t-\sqrt{r_{1}}}{m}\right\vert ^{2}}{r_{1}+\frac{%
r_{0}t^{2}-r_{1}}{m}}  \label{AB2}
\end{eqnarray}%
and
\begin{equation*}
K\left( \sigma _{0}\right) |\tilde{\phi}_{0}\rangle =\sqrt{r_{0}\left(
1-t^{2}\right) }\log m\,|e_{0}\rangle .
\end{equation*}

Motivated by the symmetry of the problem (discrete rotations around the axis
$|e_{0}\rangle $) and by the remark \ref{obs} 
we look for $\Lambda _{0}$ in the form%
\begin{equation}
\Lambda _{0}=c_{1}I+c_{0}|e_{0}\rangle \langle e_{0}|.  \label{Lambda}
\end{equation}%
Then%
\begin{eqnarray*}
\Lambda _{0}|\tilde{\phi}_{k}\rangle &=&c_{1}\sqrt{r_{1}}|e_{k}\rangle +%
\frac{\left( c_{0}+c_{1}\right) \sqrt{r_{0}}t-c_{1}\sqrt{r_{1}}}{\sqrt{m}}%
|e_{0}\rangle ;\,k\geq 1, \\
\Lambda _{0}|\tilde{\phi}_{0}\rangle &=&\left( c_{0}+c_{1}\right) \sqrt{%
r_{0}\left( 1-t^{2}\right) }|e_{0}\rangle .
\end{eqnarray*}%
Comparing coefficients in the condition (ii.b): $\Lambda _{0}|\tilde{\phi}%
_{k}\rangle =K\left( \sigma _{k}\right) |\tilde{\phi}_{k}\rangle $, we
obtain the system of equations%
\begin{eqnarray}
c_{1}\sqrt{r_{1}} &=&A(t),  \notag \\
\left( c_{0}+c_{1}\right) \sqrt{r_{0}}t-c_{1}\sqrt{r_{1}} &=&mB(t),
\label{sys} \\
\left( c_{0}+c_{1}\right) \sqrt{1-t^{2}} &=&\log m\sqrt{1-t^{2}}.  \notag
\end{eqnarray}

Two cases are possible: $t=1$ and $t<1.$ In the first case the third
equation disappears while the first two turn into%
\begin{eqnarray*}
c_{1}\sqrt{r_{1}} &=&\frac{b\log b^{2}-a\log a^{2}}{\sqrt{m}}, \\
\left( c_{0}+c_{1}\right) \sqrt{r_{0}}-c_{1}\sqrt{r_{1}} &=&-\sqrt{m}b\log
b^{2},
\end{eqnarray*}%
where we have used (\ref{AB1}), (\ref{AB2}) and (\ref{param}). Taking into
account the relations (\ref{param}), we have $\sqrt{mr_{1}}=a-b,$ and the
solution is%
\begin{eqnarray}
c_{1} &=&\frac{b\log b^{2}-a\log a^{2}}{a-b},  \notag \\
c_{0} &=&m\frac{ab\left( \log a^{2}-\log b^{2}\right) }{\left( a-b\right) %
\left[ a+\left( m-1\right) b\right] }.  \label{t1}
\end{eqnarray}

In the case $t<1$ the third equation in the system (\ref{sys}) reduces to $%
c_{0}+c_{1}=\log m.$ Eliminating $c_{0},c_{1}$ one obtains the equation for $%
t:$%
\begin{equation}
A(t)+mB(t)=\left( \log m\right) \sqrt{r_{0}}t.  \label{eq}
\end{equation}%
Substituting the expressions (\ref{AB1}), (\ref{AB2}) for $A(t),B(t)$ and
introducing the parameter%
\begin{equation*}
\tau =\sqrt{\frac{r_{0}}{r_{1}}}t,
\end{equation*}%
we can reduce (\ref{eq}) to the form%
\begin{equation}
\tau \log \left( m-1+\tau ^{2}\right)=\frac{m-1+\tau }{m}\log \left(
m-1+\tau \right) ^{2}+\frac{m-1}{m}\left( \tau -1\right) \log \left( \tau
-1\right) ^{2}.  \label{eqt1}
\end{equation}
Denoting $s=\tau ^{-1}$ equation (\ref{eqt1}) can be transformed into%
\begin{eqnarray*}
f(s)&\equiv &\ln \left[ 1+\left( m-1\right) s^{2}\right] \\
&-&2\frac{1+\left( m-1\right) s}{m}\ln \left[ 1+\left( m-1\right) s\right] -2%
\frac{m-1}{m}\left( 1-s\right) \ln \left\vert 1-s\right\vert =0.
\end{eqnarray*}%
At the point $s=0$ (corresponding to $\tau \rightarrow \infty $) this
function and its first and second derivatives vanish, while $f^{\prime
\prime \prime }(s)=2\left( m-1\right) \left( m-2\right) >0$ hence $f$ is
increasing and positive for small $s$. But $f(1)=-\ln m<0,$ hence there is a
(unique) root in the interval $(0,1],$ namely $s=\frac{m-2}{2\left(
m-1\right) }.$ There are no roots for $s>1$ because the function is
decreasing.

Thus the value $\tau _{a}(m)=\frac{2(m-1)}{m-2}$ suggested by (\ref{t0}) is
the unique positive solution of (\ref{eq}). Substituting this value in the
first equation of (\ref{sys}) and using the third equation we obtain%
\begin{equation}
c_{1}=\log m-\frac{m\log \left( m-1\right) }{m-2},\quad c_{0}=\frac{m\log
\left( m-1\right) }{m-2},  \label{c10}
\end{equation}%
the values which do not depend on the parameters of the pyramid $%
r_{0},r_{1}. $ The criterion for the applicability of this second case is
\begin{equation*}
t_{a}(m)\equiv \sqrt{\frac{r_{1}}{r_{0}}}\tau _{a}(m)=\sqrt{\frac{r_{1}}{%
r_{0}}}\frac{2(m-1)}{m-2}<1
\end{equation*}
which, via (\ref{r01}), is equivalent to $r_{0}>\frac{4(m-1)}{m^{2}}$, and,
via (\ref{param}), to $a^{2}<\frac{m-1}{m}.$

Correspondingly, the condition for applicability of the solution (\ref{t1})
is $a^{2}\geq \frac{m-1}{m}.$ By introducing the parameter $p$ such that
\begin{equation}  \label{pp}
p=a^{2}, \quad \frac{1-p}{m-1}=b^2,
\end{equation}
we have $b=\sqrt{\frac{1-p}{m-1}},$ leading to%
\begin{equation*}
\Lambda _{0}=c_{1}I+c_{0}|e_{0}\rangle \langle e_{0}|=-\mu _{1}(p)I+\mu
_{0}(p)m|e_{0}\rangle \langle e_{0}|
\end{equation*}%
with $\mu _{1}(p),\mu _{0}(p)$ given by (\ref{mu1}), (\ref{mu0}) under the
condition that $p\geq \frac{m-1}{m}$. Introducing complex numbers $%
z_{j}=\langle e_{j}|\psi \rangle ,$ so that $\sum_{j=1}^{m}\left\vert
z_{j}\right\vert ^{2}=1,$ the optimality condition (ii.a) amounts to the
inequality%
\begin{equation*}
-\sum_{j=1}^{m}\left\vert z_{j}\right\vert ^{2}\log \left\vert
z_{j}\right\vert ^{2}\geq \mu _{0}(p)\left\vert
\sum_{j=1}^{m}z_{j}\right\vert ^{2}-\mu _{1}(p).  \label{basicc}
\end{equation*}%
Since $\mu _{0}(p)\geq 0,$ this is equivalent to (\ref{basic}) where $%
t_{j}=\left\vert z_{j}\right\vert ^{2}$ .

Remarkably, inserting the limiting value $p=\frac{m-1}{m}$ we obtain for $%
m>2 $
\begin{equation*}
\mu _{1}=\frac{m\log \left( m-1\right)}{m-2} -\log m,\quad \mu _{0}=\frac{%
\log \left( m-1\right)}{m-2} ,
\end{equation*}%
which correspond to the values $c_{0},c_{1},$ given by (\ref{c10}) and
obtained by solving (\ref{sys}) in the case $t_{a}(m)<1.$

Finally, from (\ref{Lambda}) and (\ref{rhobar}) $\mathrm{Tr\,}\Lambda _{0}%
\bar{\rho}=c_{1}+c_{0}r_{0}$ so that
\begin{equation}  \label{access}
A(\mathcal{E})=\log m-c_{1}-c_{0}r_{0}.
\end{equation}
Thus we have the following corollary to the theorem \ref{th2}:

\begin{Cor}
\label{cor1} For ensemble of acute pyramids the information-optimal
observable is given by the equations (\ref{etilde}), (\ref{t0}) as
conjectured in \cite{eng}. The accessible information is given by the
expressions
\begin{equation}  \label{accinf}
A(\mathcal{E})=\left\{%
\begin{array}{cc}
\log m + p\log p +\left( 1-p\right)\log \frac{1-p}{m-1}, & p\in \left[\frac{%
m-1}{m}, 1\right] \\
\frac{m-1}{m-2}\log (m-1)\left( \sqrt{p}-\sqrt{\frac{1-p}{m-1}}\right)^2, &
p\in \left[\frac{1}{m}, \frac{m-1}{m}\right]%
\end{array}%
\right.
\end{equation}
which agree with (\ref{acc}).
\end{Cor}

\begin{proof} 
For accessible information we use the relation (\ref{access}). A routine calculation using
(\ref{t1}), (\ref{param}), (\ref{pp}) gives
(\ref{isrm1}) which is the first line in (\ref{accinf}) for moderately acute pyramids.
For strongly acute pyramid, using (\ref{c10}), one gets
\begin{equation*}
A(\mathcal{E})=\frac{m\left( 1-r_{0}\right) }{m-2}\log (m-1),
\end{equation*}
which is the second line, taking into account  (\ref{param1}).
\end{proof}

\subsection{The case of obtuse pyramids}

\label{obtus}

When $r_{0}<r_{1}$ $(b<0)$ the pyramid is obtuse, with extreme flat case $%
r_{0}=0$ when all the state vectors $|\psi _{j}\rangle $ lie in the
hyperplane orthogonal to $|e_{0}\rangle .$ The hypothetical optimal
observable is then drastically different. To describe it let us first
briefly survey the flat pyramid in the case $m=3$ where $|\psi _{j}\rangle
;j=1,2,3,$ are essentially the \textquotedblleft trine\textquotedblright\ of
equiangular unit vectors in two-dimensional real plane, cf. \cite{ippi}.
Then the optimal observable is $\mathcal{M}=\left\{ M_{k};k=1,2,3\right\} ,$
where $M_{k}=\frac{2}{3}|\psi _{k}^{\perp }\rangle \langle \psi _{k}^{\perp
}|,$ and $|\psi _{k}^{\perp }\rangle $ are the unit vectors in the plane
such that $\langle \psi _{k}^{\perp }|\psi _{k}\rangle =0;k=1,2,3$ \cite{sas}%
. In other words, $\mathcal{M}$ is \textquotedblleft unambiguous
discrimination\textquotedblright\ measurement (MUD) for the states $\rho
_{j}=|\psi _{j}\rangle \langle \psi _{j}|,\,j=1,2,3.$ In the section %
\ref{flat} we perform the analysis based on our optimality criterion showing
that the conditions (ii) of the theorem are fulfilled with $\Lambda _{0}=I,$
with the condition (ii.a) equivalent to the entropy inequality%
\begin{equation}
-\sum_{j=1}^{3}t_{j}\log t_{j}\geq 1;\quad t_{j}=\frac{2}{3}\cos ^{2}\left(
\alpha +\frac{2\pi j}{3}\right) ,  \label{trineq}
\end{equation}%
which is saturated for $\alpha =\pi /2+\pi k$.

Most interesting is the case of \textquotedblleft nearly
flat\textquotedblright\ trine when $0<r_{0}<\bar{r}_{0}\approx 0,0614.$
Then, as shown by P. Shor in \cite{shor}, the conjectured optimal observable
has 6 outcomes and is a statistical mixture of the MUD described above and
the SRM $M_{k}=|{e}_{k}\rangle \langle {e}_{k}|,\,k=1,2,3.$ The SRM becomes
optimal for obtuse pyramids with $\bar{r}_{0}\leq r_{0}\leq r_{1}=1/3.$

A generalization of this picture to the case $m\geq 3$ based on
numerical studies was suggested and elaborated in \cite{eng}. Let us apply
our optimality criterion to derive the corresponding entropy inequalities
and to obtain additional insight. We consider the ensemble $\mathcal{E}%
=\left\{ \pi _{j},\rho _{j}\right\} _{j=1,\dots ,m}$ with $\pi _{j}=\frac{1}{%
m},\,\rho _{j}=|\psi _{j}\rangle \langle \psi _{j}|,$ where $|\psi
_{j}\rangle $ are given by (\ref{pyr}) with $r_{0}<r_{1}$ $(b<0)$ . The
conjectured optimal observable $\mathcal{M}$ \cite{eng} has $m+\frac{m(m-1)}{%
2}=\frac{m(m+1)}{2}$ components of the form:
\begin{eqnarray}  \label{optm}
M_{k} &=&t^{-2}|\tilde{e}_{k}\rangle \langle \tilde{e}_{k}|;\quad k=1,\dots
,m; \\
M_{rs} &=&\frac{2}{m}\left( 1-t^{-2}\right) |rs\rangle \langle rs|;\quad
1\leq r<s\leq m,  \notag
\end{eqnarray}%
where $t\geq 1,$ the vectors $|\tilde{e}_{k}\rangle $ are given by (\ref%
{etilde}), and%
\begin{equation*}
|rs\rangle =\frac{1}{\sqrt{2mr_{1}}}\left( |\psi _{r}\rangle -|\psi
_{s}\rangle \right) =\frac{1}{\sqrt{2}}\left( |e_{r}\rangle -|e_{s}\rangle
\right) .
\end{equation*}%
Therefore the dual ensemble $\mathcal{E}^{\prime }$ consists of the states
(we use the same notation as for similar objects in the acute case):%
\begin{equation*}
\sigma _{k}=p_{k}^{-1}|\tilde{\phi}_{k}\rangle \langle \tilde{\phi}%
_{k}|\,,\quad p_{k}=\langle \tilde{\phi}_{k}|\tilde{\phi}_{k}\rangle ,
\end{equation*}%
where
\begin{eqnarray*}
|\tilde{\phi}_{k}\rangle &=&t^{-1}\overline{\rho }^{1/2}|\tilde{e}%
_{k}\rangle =t^{-1}\left( \sqrt{r_{1}}|e_{k}\rangle +\frac{\sqrt{r_{0}}t-%
\sqrt{r_{1}}}{\sqrt{m}}|e_{0}\rangle \right), \\
p_{k} &=&t^{-2}\left( r_{1}+\frac{r_{0}t^{2}-r_{1}}{m}\right) ;\quad
k=1,\dots ,m,
\end{eqnarray*}%
and%
\begin{equation*}
\sigma _{rs}=p_{rs}^{-1}|\tilde{\phi}_{rs}\rangle \langle \tilde{\phi}%
_{rs}|\,,\quad p_{rs}=\langle \tilde{\phi}_{rs}|\tilde{\phi}_{rs}\rangle ,
\end{equation*}%
where%
\begin{eqnarray*}
|\tilde{\phi}_{rs}\rangle &=&\sqrt{\frac{2}{m}\left( 1-t^{-2}\right) }\,%
\overline{\rho }^{1/2}|rs\rangle =\sqrt{\frac{r_{1}}{m}\left(
1-t^{-2}\right) }\,\left( |e_{r}\rangle -|e_{s}\rangle \right) ; \\
p_{rs} &=&\frac{2r_{1}}{m}\left( 1-t^{-2}\right) .
\end{eqnarray*}%
To check the optimality condition (ii.b) we proceed in parallel with the
acute case, and compute%
\begin{equation*}
K\left( \sigma _{k}\right) |\tilde{\phi}_{k}\rangle =t^{-1}\left(
A(t)|e_{k}\rangle +B(t)\sqrt{m}|e_{0}\rangle \right) ;\quad k=1,\dots ,m,
\end{equation*}%
where $A(t),B(t)$ are given by (\ref{AB1}), (\ref{AB2}). Also%
\begin{equation}
K\left( \sigma _{rs}\right) |\tilde{\phi}_{rs}\rangle
=-\sum_{j}|e_{j}\rangle \langle e_{j}|\tilde{\phi}_{rs}\rangle \,\log
p_{rs}^{-1}\mathrm{\,}\,\left\vert \langle e_{j}|\tilde{\phi}_{rs}\rangle
\right\vert ^{2},  \label{Krs}
\end{equation}%
with%
\begin{eqnarray*}
\langle e_{j}|\tilde{\phi}_{rs}\rangle &=&\sqrt{\frac{r_{1}}{m}\left(
1-t^{-2}\right) }\left( \delta _{jr}-\delta _{js}\right) , \\
\left\vert \langle e_{j}|\tilde{\phi}_{rs}\rangle \right\vert ^{2} &=&\frac{%
r_{1}}{m}\left( 1-t^{-2}\right) \left( \delta _{jr}+\delta _{js}\right) .
\end{eqnarray*}%
Substituting this into (\ref{Krs}) results in%
\begin{equation}
K\left( \sigma _{rs}\right) |\tilde{\phi}_{rs}\rangle =|\tilde{\phi}%
_{rs}\rangle ,\quad 1\leq r<s\leq m.  \label{chudo}
\end{equation}%
Looking again for $\Lambda _{0}=c_{1}I+c_{0}|e_{0}\rangle \langle e_{0}|,$
and comparing coefficients in the condition (ii.b): $\Lambda _{0}|\tilde{\phi%
}_{k}\rangle =K\left( \sigma _{k}\right) |\tilde{\phi}_{k}\rangle $, we
obtain the first two equations of the system (\ref{sys}), namely%
\begin{eqnarray}
c_{1}\sqrt{r_{1}} &=&A(t),  \label{sys1} \\
\left( c_{0}+c_{1}\right) \sqrt{r_{0}}t-c_{1}\sqrt{r_{1}} &=&mB(t),  \notag
\end{eqnarray}%
while $\Lambda _{0}|\tilde{\phi}_{rs}\rangle =K\left( \sigma _{rs}\right) |%
\tilde{\phi}_{rs}\rangle $ via (\ref{chudo}) amounts to
\begin{equation}
\left( 1-t^{-2}\right) \left( c_{1}-1\right) =0,  \label{sys2}
\end{equation}%
since $\langle e_{0}|\tilde{\phi}_{rs}\rangle \equiv 0.$

Again, two cases are possible: $t=1$ and $t>1.$ In the first case the third
equation disappears while the first two lead to the same formulas (\ref{t1})
as in the acute case, however with $a=\sqrt{p\text{ }}$, but $b=-\sqrt{\frac{%
1-p}{m-1}}.$ Then, denoting $z_{j}=\langle e_{j}|\psi \rangle $, the
condition (ii.a) brings us to the inequality:

\textit{For complex $z_{j}$ satisfying }$\mathit{\sum_{j=1}^{m}\left\vert
z_{j}\right\vert ^{2}=1,}$\textit{$\,$it holds
\begin{equation}
-\sum_{j=1}^{m}\left\vert z_{j}\right\vert ^{2}\log \left\vert
z_{j}\right\vert ^{2}\geq \tilde{\mu}_{0}(p)\left\vert
\sum_{j=1}^{m}z_{j}\right\vert ^{2}-\tilde{\mu}_{1}(p),  \label{basic2}
\end{equation}%
where
\begin{equation}
\tilde{\mu}_{0}(p)=\frac{\sqrt{\frac{p\left( 1-p\right) }{m-1}}\left( \log
p-\log \frac{1-p}{m-1}\right) }{\left( \left( m-1\right) \sqrt{\frac{1-p}{m-1%
}}-\sqrt{p}\right) \left( \sqrt{p}+\sqrt{\frac{1-p}{m-1}}\right) },
\label{0tilde}
\end{equation}%
\begin{equation}
\tilde{\mu}_{1}(p)=\frac{\sqrt{p}\log p+\sqrt{\frac{1-p}{m-1}}\log \frac{1-p%
}{m-1}}{\left( \sqrt{p}+\sqrt{\frac{1-p}{m-1}}\right) }.  \label{1tilde}
\end{equation}%
The equality in (\ref{basic2}) is attained for $z_{1}=\sqrt{p},\,z_{j}=-%
\sqrt{\frac{1-p}{m-1}},\,j=2,\dots ,m,$ and all the permutations and total change of phase of these $%
z_{j}$.}

Note that here $\tilde{\mu}_{0}(p)\leq 0,$ which prevents us from passing to
the probabilities $t_{j}=\left\vert z_{j}\right\vert ^{2}$. However, as the
lemma \ref{lemma} shows, we can always restrict to the case of real $z_{j}$.

In (\ref{basic2}) the parameter $p$ varies in the range $[\max \left\{ \frac{%
m-1}{m},p(m)\right\} ,1],$ where $p(m)$ is found below. The accessible
information in this case is given by the same expression (\ref{isrm1}) as in
the case of moderately acute pyramids. For future use we note that
\begin{equation}
-\tilde{\mu}_{1}\left( \frac{m-1}{m}\right) =\log m-\frac{m-2}{m}\log (m-1).
\label{val}
\end{equation}

If $t>1,$ then (\ref{sys2}) imply $c_{1}=1$ and equations (\ref{sys1})
reduce to
\begin{eqnarray}
\sqrt{r_{1}} &=&A(t),  \notag \\
c_{0} &=&\frac{m}{\sqrt{r_{0}}t}B(t)+\frac{1}{t}\sqrt{\frac{r_{1}}{r_{0}}}-1.
\label{c0}
\end{eqnarray}%
Here the first equation serves for determination of $t$ (compatible with the
condition $t>1$). By using (\ref{AB1}) it can be transformed to%
\begin{equation}
\log \frac{m\left( m-1+\tau ^{2}\right) }{2}-\frac{m-1+\tau }{m}\log \left(
m-1+\tau \right)^2 -\frac{\left( 1-\tau \right) }{m}\log \left( 1-\tau
\right)^2 =0,  \label{eqt2}
\end{equation}%
where $\tau =\sqrt{\frac{r_{0}}{r_{1}}}t.$

Denoting $g(\tau )$ the function in equation (\ref{eqt2}), we have%
\begin{equation*}
g(0)=\log \frac{m}{2}-\frac{m-2}{m}\log \left( m-1\right) \left\{
\begin{array}{cc}
>0, & m\leq 6; \\
<0, & m\geq 7,%
\end{array}%
\right. .
\end{equation*}%
The function is decreasing in the interval $[0,1]$ and $g(1)=-\log 2<0$,
hence it has a root $\tau_o(m)$ in this interval if and only if $m\leq 6.$
For $\tau >1$ the function is negative and has no roots.

The condition $t>1$ amounts to $\tau _{o}(m)>\sqrt{\frac{r_{0}}{r_{1}}},$
which via (\ref{r01}), is equivalent to
\begin{equation*}
r_{0}<\left( 1+\frac{m-1}{\tau _{o}(m)^{2}}\right) ^{-1}\equiv r_{0}(m),
\end{equation*}%
and, via (\ref{param}) and (\ref{pp}), to
\begin{eqnarray}
p\equiv a^{2}&<&\left[ 1+\left( m-1\right) \left( 1-\frac{m}{1-\tau _{o}(m)}
\right) ^{-2}\right] ^{-1}  \notag \\
&=&\frac{(m-1+\tau _{o}(m))^2}{m(m-1+\tau _{o}(m)^2)}\equiv p(m).
\label{cond}
\end{eqnarray}
One can check that equation (\ref{eqt2}) with $\tau=\tau _{o}(m)$ is
equivalent to
\begin{equation}  \label{val1}
-\tilde{\mu}_1(p(m))=1.
\end{equation}

The numerically found values\footnote{%
By using equation (\ref{eqt2}) it is possible to express the function $p(m)$
via the Lambert $W$-function.}:%
\begin{equation*}
\begin{tabular}{|l|l|l|l|l|l|l|}
\hline
$m$ & 2 & 3 & 4 & 5 & 6 & 7 \\ \hline
$\tau _{o}(m)$ & 0 & 0,3616 & 0,25866 & 0,15195 & 0,04781 & [-0,05286] \\
\hline
$mr_{0}(m)$ & 0 & 0,1841 & 0,08726 & 0,02869 & 0,00274 & -- \\ \hline
$p(m)$ & 0,5 & 0,8725 & 0,8656 & 0,85699 & 0,84896 & 0,8417 \\ \hline
$\frac{m-1}{m}$ & 0,5 & 0,6(6) & 0,75 & 0,8 & 0.83(3) & 0,8571 \\ \hline
\end{tabular}%
\end{equation*}%
Here $mr_{0}(m)$ show a good agreement with the values found empirically for
$m\leq 6$ in \cite{eng}, eq. (39). In the case $m=7$, as well as for greater
values of $m$, the solution of the equation (\ref{eqt2}) becomes negative
and also $p(m)<\frac{m-1}{m}.$ This makes plausible the hypothesis that for $%
m\geq 7 $ the optimal observable is SRM, and the corresponding entropy
inequality is (\ref{basic2}) for all $p\geq \frac{m-1}{m}.$

Then, for $m\leq 6$, under the condition (\ref{cond}) the equation (\ref%
{sys2}) implies $c_{1}=1$ while from (\ref{c0})%
\begin{equation*}
c_{0}=\left( \tau _{o}(m)^{-1}-1\right) \log \frac{2\left( \tau
_{o}(m)-1\right) ^{2}}{m\left( m-1+\tau _{o}(m)^{2}\right) }.
\end{equation*}%
The condition (ii.a) then amounts to the entropy inequality (\ref{basic2})
with $\tilde{\mu}_{0}=c_{0}/m,\quad \tilde{\mu}_{1}=-1$ i.e.
\begin{equation}
-\sum_{j=1}^{m}\left\vert z_{j}\right\vert ^{2}\log \left\vert
z_{j}\right\vert ^{2}\geq \frac{\left( \tau _{o}(m)^{-1}-1\right) }{m}\log
\frac{2\left( \tau _{o}(m)-1\right) ^{2}}{m\left( m-1+\tau
_{o}(m)^{2}\right) }\left\vert \sum_{j=1}^{m}z_{j}\right\vert ^{2}+1.
\label{ineqz}
\end{equation}%
Due to lemma \ref{lemma} one can restrict here to real $z_{j}$ , $%
\sum_{j=1}^{m}\left\vert z_{j}\right\vert ^{2}=1.$ The equality is attained
for $z_{1}=-z_{2}=1/\sqrt{2},z_{j}=0$, $j\geq 3$; for $z_{1}=\sqrt{p(m)}%
,\,z_{j}=-\sqrt{\frac{1-p(m)}{m-1}},\,j\geq 2$; and for all permutations
and total change of sign of
these $z_{j}.$

In this case the accessible information, by using (\ref{access}), is
\begin{eqnarray}  \label{assobt}
A(\mathcal{E}) &=&\log \frac{m}{2}-r_{0}\left( \tau _{o}(m)^{-1}-1\right)
\log \frac{2\left( \tau _{o}(m)-1\right) ^{2}}{m\left( m-1+\tau
_{o}(m)^{2}\right) } \\
&=& \log\dfrac{m}{2}-r_0m\tilde{\mu}_0(p(m)).  \notag
\end{eqnarray}

The following theorem summarizes the content of this section.

\begin{Th}
\label{th3} If $m\geq 7$ then $p(m)<\frac{m-1}{m}$ and the entropy
inequality (\ref{basic2}) holds for $p\in \left[\frac{m-1}{m}, 1\right]$. If
$m\leq 6$ then $p(m)>\frac{m-1}{m}$ and the inequality (\ref{basic2}) holds
for $p\in \left[p(m), 1\right]$, yielding the inequality (\ref{ineqz}) for $%
p\in \left[\frac{m-1}{m}, p(m)\right]$.
\end{Th}

The suggestion for a proof of the inequalities combining analytical argument with
computer-aided reasoning is given in section \ref{sec:obtuse}.

\begin{Cor}
\label{cor2} In the case of obtuse pyramids the information-optimal
observable is given by the relations (\ref{optm}) with $t=1$ for $m\geq 7$,
and $t=\sqrt{\frac{r_1}{r_0}}\tau_{o}(m)$ for $m\leq 6$. The accessible
information is given by the expression (\ref{isrm1}) if $p\in [\max \left\{
\frac{m-1}{m},p(m)\right\} ,1],$ otherwise by (\ref{assobt}) if $m\leq 6$
and $p\in \left[\frac{m-1}{m}, p(m)\right]$.
\end{Cor}

\subsection{The case of flat pyramid}

\label{flat}

The considerations above formally do not include the case of flat pyramid
with $r_{0}=0,r_{1}=1/(m-1).$ Then the edges (\ref{pyr}) of the pyramid
\begin{equation*}
|\psi _{j}\rangle =\sqrt{\frac{m-1}{m}}|e_{j}\rangle -\sqrt{\frac{1}{m}}%
|e_{0}\rangle ;\quad j=1,\dots ,m
\end{equation*}%
all lie in the hyperplane $\mathcal{H}_{0}=\left\{ \psi :\langle e_{0}|\psi
\rangle =0\right\} $. The vectors are linearly dependent and their convex
hull is a regular symplex in $\mathcal{H}_{0}.$ The average state of
the ensemble $\mathcal{E}$
\begin{equation*}
\overline{\rho }=\frac{1}{m-1}\left( I-|e_{0}\rangle \langle e_{0}|\right)
\end{equation*}%
is degenerate and we should use generalized inverse in the formula (\ref%
{dualobs}), resulting in%
\begin{equation*}
M_{j}^{\prime }=|\tilde{\psi}_{j}\rangle \langle \tilde{\psi}_{j}|,\quad
j=1,\dots ,m,
\end{equation*}%
where
\begin{equation}
|\tilde{\psi}_{j}\rangle =|e_{j}\rangle -\sqrt{\frac{1}{m}}|e_{0}\rangle =%
\sqrt{\frac{m-1}{m}}|\psi _{j}\rangle.   \label{psitilde}
\end{equation}%
Denoting $z_{j}=\langle e_{j}|\psi \rangle $ we have
\begin{equation}
\sum_{j=1}^{m}\left\vert z_{j}\right\vert ^{2}=1,\quad \sum_{j=1}^{m}z_{j}=0.
\label{condw}
\end{equation}%
Indeed, the unit vector $|\psi \rangle $ lies in the hyperplane $\langle
e_{0}|\psi \rangle =0$i equality follows.

Arguing as in the previous section, we suggest that the optimal observable
is given by (\ref{optm}) with $t=1$ for $m\geq 7$ and $t=\infty $ for $m\leq
6$. The corresponding tight entropy inequality is
\begin{equation}
-\sum_{j=1}^{m}\left\vert z_{j}\right\vert ^{2}\log \left\vert
z_{j}\right\vert ^{2}\geq \left\{
\begin{array}{cc}
1, & m\leq 6; \\
\log m-\frac{m-2}{m}\log (m-1), & m\geq 7,%
\end{array}%
\right.  \label{ineqw}
\end{equation}%
for complex $z_{j}$ satisfying (\ref{condw}) \ (due to a simple modification
of lemma \ref{lemma} one can restrict here to real $z_{j}$ ). Indeed, for $%
m\leq 6$ we can use the argument leading to inequality (\ref{ineqz}) with $%
\sum_{j=1}^{m}z_{j}=0$ and $p=p(m)$. Then (\ref{val1}) implies the first
line in the right-hand side of (\ref{ineqw}). On the other hand, for $m\geq
7 $ we should use the inequality (\ref{basic2}) with $\sum_{j=1}^{m}z_{j}=0$
and $p=\frac{m-1}{m}.$ Then the first term in the righthand side vanish,
while $-\tilde{\mu}_{1}(p)\ $ gets the value (\ref{val}) 
i.e. the second line in the righthand side of (\ref{ineqw}), which is
greater than 1 for $m\leq 6$, but less for $m\geq 7$. We thus arrive at the
following special case of theorem \ref{th3}:%


\begin{prop}
\label{prop1} The tight entropy inequality (\ref{ineqw}) holds for real (or
complex) $z_{j}$ satisfying the conditions (\ref{condw}).

In the case $m\leq 6$ the equality in (\ref{ineqw}) is attained for $%
z_{1}=-z_{2}=1/\sqrt{2},\,z_{j}=0$ for $j\geq 3;$ in the case $m\geq 7$ --
for $z_{1}=\sqrt{\frac{m-1}{m}},\,z_{j}=-\sqrt{\frac{1}{\left( m-1\right) m}}%
,\,j\geq 2,$ (and for all permutations and total change of phase of such $z_{j}).$
\end{prop}

Theorem \ref{th1} then implies

\begin{Cor}
\label{cor3} For the ensemble of flat pyramid, information-optimal
observable is given by (\ref{optm}) with $t=1$ for $m\geq 7$ and $t=\infty$
for $m\leq 6$.

The corresponding accessible information is%
\begin{equation*}
A(\mathcal{E})=\log m-\mathrm{Tr\,}\Lambda _{0}\bar{\rho}=\left\{
\begin{array}{cc}
\log \frac{m}{2}, & m\leq 6; \\
\frac{m-2}{m}\log (m-1), & m\geq 7.%
\end{array}%
\right.
\end{equation*}
\end{Cor}

For $m=3$ the inequality (\ref{ineqw}) reduces to (\ref{trineq}). Indeed, as
shown in lemma \ref{lemma}, one can restrict to real $z_{j}$
 in (\ref{ineqw}). Then $|\psi \rangle $ is a vector in the real plane $\mathcal{H}_{0}$ and
by (\ref{psitilde})
\begin{equation*}
z_{j}=\langle e_{j}|\psi \rangle =\langle \tilde{\psi}_{j}|\psi \rangle =%
\sqrt{2/3}\langle \psi _{j}|\psi \rangle =\sqrt{2/3}\cos \alpha _{j},\quad
j=1,2,3,
\end{equation*}%
where $\alpha _{j}$ is the angle between the vector $|\psi \rangle $ and the
ensemble vector  $|\psi _{j}\rangle $ in the plane $\mathcal{H}_{0}.$ The
angle between any two different vectors of the ensemble is $2\pi /3$ because  $%
\langle \psi _{j}|\psi _{k}\rangle =-1/2$ according to (\ref{cosine}).
Therefore $\alpha _{j}=\alpha +\frac{2\pi j}{3},\quad j=1,2,3,$ thus we come
to (\ref{trineq}).

Let us also explain how (\ref{trineq}) follows from the result of the paper \cite{sas}%
. A simple calculation shows that  (\ref{trineq}) is equivalent to%
\begin{equation*}
\log 3-\frac{1}{3}\sum_{j=1}^{3}\left( 1+\cos \left( 2\alpha +\frac{4\pi j}{3%
}\right) \right) \log \left( 1+\cos \left( 2\alpha +\frac{4\pi j}{3}\right)
\right) \geq \log 2;
\end{equation*}%
A special case of lemma 3 in \cite{sas} (with $M=3$) implies that the
deductible in the lefthand side  is maximized for $\alpha =\pi /2+\pi k
$ (giving in fact the accessible information $\log \left( 3/2\right) $). The
proof in \cite{sas} is based on the decomposition of  the function $\left( 1+x\right) \log
\left( 1+x\right) $ into the power series and a clever manipulation with the
terms of the decomposition.

\section{Suggestions for proofs}

\label{prrrofs}

\subsection{The case of theorem \protect\ref{th2}}

\label{prroof}

Let $m\geq 2,\ \dfrac{m-1}{m}\leq p<1$ and let $\mu _{0}(p),\mu _{1}(p)$ be
given by (\ref{mu0}), (\ref{mu1}). 


\textit{For $p\in \lbrack (m-1)/m,1)$ and a probability distribution $%
(t_{1},...,t_{m})$, it holds
\begin{equation}
F(t)\equiv H(t)-\Bigl(\mu _{0}(p)\Bigl(\sum\limits_{k=1}^{m}\sqrt{t_{k}}%
\Bigl)^{2}-\mu _{1}(p)\Bigl)\geq 0,  \label{func_F}
\end{equation}%
where $H(t)=-\sum\limits_{k=1}^{m}t_{k}\log t_{k}$ is the Shannon entropy.}

In the special case $p=1$, the inequality \eqref{func_F} is trivially
satisfied by substituting $\mu_0(1)\equiv0,\ \mu_1(1)=0$. 

\begin{proof}  The function $F(t)$ cannot reach a minimum on the boundary of the
simplex of probability distributions, since the function 
\begin{equation*}
u(s)= -\sum\limits_{i=1}^m
(t_{*,i}+sa_i)\log(t_{*,i}+sa_i)-\mu_0(p)\Bigl(\sum\limits_{i=1}^m\sqrt{%
t_{*,i}+sa_i}\Bigl)+\mu_1(p)
\end{equation*}
for a boundary point $t_*$ and a transversal direction $a$ has the
derivative $-\infty $ at zero:
\begin{equation*}
u^{\prime}(0)=\lim\limits_{t\to t_*}\Bigg[-\sum\limits_{i\in I_0} a_i\log
t_{i} - C\Bigl(\sum\limits_{i\in I_0} \dfrac{a_i}{\sqrt{t_i}}\Bigl)+O(1)\Bigg],
\end{equation*}
where $I_0=\{1\leq j\leq m;\ t_{*,j}=0\}$, $C>0$ is a
constant. In this case, $a_i>0,\ i\in I_0$ (direction inside the simplex),
hence $u^{\prime}(0)=-\infty$.\end{proof}

Thus it is sufficient to check the inequality \eqref{func_F} only at the
critical points.

\begin{Lem}
\label{ac_Lem2} The critical points $t_*=(t_{*,1},...,t_{*,m})$ of the
function $F(t)$ 
have the form
\begin{equation}  \label{crit}
t_*=\Bigl(q/k,...,q/k,(1-q)/(m-k),...,(1-q)/(m-k)\Bigl)
\end{equation}
up to permutations, for some integer $k$, $1\leq k\leq m $ and $q\in[0,1]$
(for $k=m$ it is assumed that $q=1/m$). One of the critical points is $%
(p,(1-p)/(m-1),...,(1-p)/(m-1))$, at which the value of the function $F(t)$
is equal to $0$.
\end{Lem}

\begin{proof}  Consider the function  $%
u(s)=F(t+sa)$, where $t=(t_1,...,t_m),\ a=(a_1,...,a_m),\ \sum\limits_i a_i=0
$, for a range of $s$ such that $t+as$ defines a probability
distribution. Then
\begin{align*}
\dfrac{d}{ds}u(s)& =-\sum\limits_{i=1}^{m}a_{i}\log t_{i}-\mu _{0}(p)\Bigl(\sum\limits_{i=1}^{m}\dfrac{a_{i}}{\sqrt{t_{i}}}\Bigl)\Bigl(\sum\limits_{i=1}^{m}\sqrt{t_{i}}\Bigl) \\
& =\sum\limits_{i=2}^{m}a_{i}\Bigg[-\log t_{i}+\log t_{1}-\mu _{0}(p)\Bigl(\sum\limits_{k=1}^{m}\sqrt{t_{k}}\Bigl)\Bigl(\dfrac{1}{\sqrt{t_{i}}}-\dfrac{%
1}{\sqrt{t_{1}}}\Bigl)\Bigg].
\end{align*}

The condition $\dfrac{d}{ds}u(0)=0$ for a point $t_*=t$ and an arbitrary
vector $a$ (that is, the point $t_*$ is critical) is satisfied if and only
if for any $2\leq i\leq m$ it holds
\begin{equation*}
\log t_{*,1}-\log t_{*,i}=\mu_0(p)\Bigl(\sum\limits_{k=1}^m \sqrt{t_{*,k}}\Bigl)\Bigl(%
\dfrac{1}{\sqrt{t_{*,i}}}-\dfrac{1}{\sqrt{t_{*,1}}}\Bigl).
\end{equation*}
Let us denote $C_1=\mu_0(p)\Bigl(\sum\limits_{k=1}^m \sqrt{t_{*,k}}\Bigl), C_2=\log t_{*,1}+\dfrac{C_1}{\sqrt{t_{*,1}}} $. The equation
\begin{equation}\label{ac_Lem2_eq1}
C_2-\log s=C_1/\sqrt{s},\quad s>0,
\end{equation}
has at most two roots, since the derivatives of the right and the left parts with respect to $s$ coincide only
when $s=C_1^2/4$. But all the numbers $t_{*,1},...,t_{*,m}$ are solutions of the equation \eqref{ac_Lem2_eq1}. Thus, among the values of $\{t_{*,i}\}$ there are no more than
two different values, hence the lemma follows. \end{proof}

At the critical points \eqref{crit} the inequality $F(t_{*})\geq 0$ reduces
to
\begin{eqnarray*}
&&\mu _{0}(p,1/m)\Bigl(\sqrt{q(k/m)}+\sqrt{(1-q)(1-k/m)}\Bigl)^{2}-\mu
_{1}(p,1-1/m)\quad \\
&\leq &-q\log \dfrac{q}{k/m}-(1-q)\log \dfrac{1-q}{1-k/m},
\end{eqnarray*}%
where $\mu_0(p,\alpha)$ and $\mu_1(p,\alpha)$ are defined as
\begin{align*}
\mu_0(p,\alpha)=\dfrac{\sqrt{\dfrac{p(1-p)}{\alpha(1-\alpha)}}\Bigl(\log%
\dfrac{p}{\alpha}-\log\dfrac{1-p}{1-\alpha}\Big)}{\Bigl(\sqrt{\alpha p}+%
\sqrt{(1-\alpha)(1-p)}\Bigl)\Bigl(\sqrt{\dfrac{p}{\alpha}}-\sqrt{\dfrac{1-p}{%
1-\alpha}}\Bigl)},  
\\
\mu_1(p,\alpha)=\dfrac{\sqrt{\dfrac{p}{\alpha}}\log\dfrac{p}{\alpha}-\sqrt{%
\dfrac{1-p}{1-\alpha}}\log\dfrac{1-p}{1-\alpha}}{\sqrt{\dfrac{p}{\alpha}}-%
\sqrt{\dfrac{1-p}{1-\alpha}}}  
\end{align*}
for real variables $p$ and $\alpha$ varying in the range $(0,1)$.

But this inequality is a special case of lemma \ref{ac_Lem3} below, the
proof of which is based on the behavior of the difference between binary
entropy and a function containing both a linear component and a square root
term $\sqrt{t(1-t)}$. We perform a qualitative analysis of the minimizer of
this function. The global minimum (there can be no more than two local
minima) appears on either the interval $[0, 1/2]$ or $[1/2, 1]$.
Importantly, this segment contains no other critical points. Therefore,
minimizing the function reduces to identifying the critical point within the
appropriate interval and evaluating the function at that point. A detailed
proof of this statement follows from

\begin{Lem}
\label{ac_Lem4} Let $A> 0,\ B,\ C\in \mathbb{R}$ (the notation of the
coefficients is used only within the framework of this lemma). Then the
function
\begin{equation}  \label{ac_Lem4_form}
g(t)=h(t)-(A\sqrt{t(1-t)}+Bt+C)
\end{equation}
reaches a local minimum on $[0,1]$ at no more than two points $t_1,t_2:\
0<t_1\leq 1/2\leq t_2<1$ which are critical. In this case, if $B>0$ and if $%
t_*>1/2$ is critical (or $t_*<1/2$ and $B< 0$), then $t_*$ is a point of the
global minimum.
\end{Lem}

\begin{proof} We have
\begin{equation*}
g^{\prime}(0)=-\infty,\quad g^{\prime}(1)=+\infty,
\end{equation*}
therefore local minima are on the interval $(0,1)$.\newline

In the calculations we will pass to the natural logarithm.
\begin{align*}
g^{\prime}(t)=-\ln\dfrac{t}{1-t}-\dfrac{A\sqrt{t(1-t)}}{2}\Bigl(\dfrac{1}{t}-%
\dfrac{1}{1-t}\Bigl)-B= \\
=-\ln\dfrac{t}{1-t}-\dfrac{A(1-2t)}{2\sqrt{t(1-t)}}-B; \\
g^{\prime\prime}(t)=-\dfrac{1}{t(1-t)}+\dfrac{A}{4(t(1-t))^{3/2}}. \\
g^{\prime\prime}(t)=0\quad \Leftrightarrow\quad A=4\sqrt{t(1-t)}.
\end{align*}
Therefore, there are two inflection points located on different sides of $%
t=1/2$. It follows that there are three critical points $t_1\leq t_0\leq t_2$,
where $t_1,\ t_2$ are local minimum points.\newline

The last statement obviously follows from the symmetric case. Indeed, for $%
B<0$, it holds $g(t)\leq g(1-t)$ for $t\in[1/2,1]$, therefore the global
minimum lies on the segment $[1/2,1]$. Let $s,\ 1-s$ ($0<s< 1/2$) be zeros
of $g^{\prime\prime}(t)$. On the segment $[s,1-s]$ the function $%
g(t)^{\prime}$ is decreasing, on the segments $[0,s],\ [1-s,1]$ the
function $g(t)^{\prime}$ increases, and due to symmetry $(g(t)+Bt)^{\prime}$
has root $1/2$. Therefore, on the segment $[1/2,1]$ when $B<0$ the function $%
g(t)$ has no more than one critical point. \end{proof}

Let us finally proceed to the proof of lemma \ref{ac_Lem3}. The variables $%
\alpha,\ p,\ \beta,\ t$ used in what follows, unless otherwise stated, are
such that
\begin{equation*}
\alpha\in{(0,1/2],\quad p\in[1-\alpha, 1)},\quad \beta\in[\alpha,1/2],\quad t%
\in[0,1].
\end{equation*}

\begin{Lem}
\label{ac_Lem3} For $\alpha ,\ p,\ \beta ,\ t$ specified above one has the
inequality
\begin{eqnarray*}
&&\mu _{0}(p,\alpha )\Bigl(\sqrt{\beta t}+\sqrt{(1-\beta )(1-t)}%
\Bigl)^{2}-\mu _{1}(p,\alpha )  \notag \\
&\leq &-t\log \dfrac{t}{\beta }-(1-t)\log \dfrac{1-t}{1-\beta }.
\label{ac_Lem3_ineq1}
\end{eqnarray*}
\end{Lem}

\begin{proof}
Let
\begin{equation*}
f(t)=-t\log \dfrac{t}{\beta }-(1-t)\log \dfrac{1-t}{1-\beta}-\Bigl[\mu _{0}(p,\alpha )\Bigl(\sqrt{\beta t}+\sqrt{(1-\beta )(1-t)}\Bigl)^{2}-\mu _{1}(p,\alpha)\Bigl],
\end{equation*}
for which $A=2\mu_0(p,\alpha)\sqrt{\beta(1-\beta)}$, $B=\log\dfrac{1-\beta}{\beta}+\mu_0(p,\alpha)(2\beta-1)$ in the form \eqref{ac_Lem4_form} of lemma \ref{ac_Lem4}.

The derivative of $f(t)$ is
\begin{equation*}
f'(t)=\log\dfrac{1-t}{t}+\log\dfrac{\beta}{1-\beta}-\mu_0(p,\alpha)\Bigg(\sqrt{\dfrac{\beta}{t}}-\sqrt{\dfrac{1-\beta}{1-t}}\Bigg)\Bigl(\sqrt{\beta t}+\sqrt{(1-\beta)(1-t)}\Bigl).
\end{equation*}

Consider the following cases:
\begin{enumerate}
\item \underline{$B\leq 0$.}

In this case there is minimum point $t_*=\beta$ ($f'(\beta)=0$) of the function $f(t)$ on $[0,1/2]$ due to lemma \ref{ac_Lem4} presented later. One has
\begin{equation}\label{ac_Lem3_ineq2}
f(\beta)=\mu_1(p,\alpha)-\mu_0(p,\alpha)\geq 0.
\end{equation}

Indeed, the inequality \eqref{ac_Lem3_ineq2} is equivalent to
\begin{eqnarray}
\Bigl(\sqrt{\dfrac{p}{\alpha}}\log\dfrac{p}{\alpha}-\sqrt{\dfrac{1-p}{1-\alpha}}\log\dfrac{1-p}{1-\alpha}\Bigl)\Bigl(\sqrt{p\alpha}+\sqrt{(1-p)(1-\alpha)}\Bigl)- \notag\\
-\sqrt{\dfrac{p(1-p)}{\alpha(1-\alpha)}}\Bigl(\log\dfrac{p}{\alpha}-\log\dfrac{1-p}{1-\alpha}\Bigl)\geq 0\quad \Leftrightarrow\notag\\
p\log\dfrac{p}{\alpha}\Bigg(1-\sqrt{\dfrac{\alpha(1-p)}{p(1-\alpha)}}\Bigg)-(1-p)\log\dfrac{1-p}{1-\alpha}\Bigg(1-\sqrt{\dfrac{p(1-\alpha)}{\alpha(1-p)}}\Bigg)\geq 0\quad \Leftrightarrow \notag\\
\sqrt{\alpha p}\log\dfrac{p}{\alpha}+\sqrt{(1-\alpha)(1-p)}\log\dfrac{1-p}{1-\alpha}\geq 0.\quad\label{ac_Lem3_ineq3}
\end{eqnarray}

Let $V(p)=\sqrt{\alpha p}\ln\dfrac{p}{\alpha}+\sqrt{(1-\alpha)(1-p)}\ln\dfrac{1-p}{1-\alpha}$, then
\begin{enumerate}
\item
\begin{eqnarray*}
(d/dp)^2 V(p)=-\dfrac{\sqrt{\alpha}\ln(p/\alpha)}{4p^{3/2}}-\dfrac{\sqrt{1-\alpha}\ln((1-p)/(1-\alpha))}{4(1-p)^{3/2}}\geq 0
\end{eqnarray*}
(since $ p\geq 1-\alpha\geq 1/2$);
\item $V(1-\alpha)=0$ and
\begin{eqnarray*}
V'(1-\alpha)=\sqrt{\dfrac{\alpha}{1-\alpha}}-1+\dfrac{1}{2}\ln\dfrac{\alpha}{1-\alpha}+\dfrac{1}{2}\sqrt{\dfrac{\alpha}{1-\alpha}}\ln\dfrac{1-\alpha}{\alpha}\geq 0\quad \Leftrightarrow\\
\dfrac{\alpha}{1-\alpha}\Bigl(1+\dfrac{1}{2}\ln\dfrac{1-\alpha}{\alpha}\Bigl)-\Bigl(1-\dfrac{1}{2}\ln\dfrac{1-\alpha}{\alpha}\Bigl)\geq 0\quad \Leftrightarrow\\
\ln\dfrac{1-\alpha}{\alpha}\geq 2-4\alpha.
\end{eqnarray*}
\end{enumerate}

Therefore, $V(p)\geq 0$ and the inequality \eqref{ac_Lem3_ineq3} is established.

\item \underline{$B\geq 0$.}

This condition implies that the function $f(t)$ has the minimum point $t_*$ on $[1/2,1]$ (lemma \ref{ac_Lem4} is applied in the same way). This point is characterized by the equation
\begin{equation*}
\log\dfrac{1-t_*}{t_*}+\log\dfrac{\beta}{1-\beta}-\mu_0(p,\alpha)\Bigg(\sqrt{\dfrac{\beta}{t_*}}-\sqrt{\dfrac{1-\beta}{1-t_*}}\Bigg)\Bigl(\sqrt{\beta t_*}+\sqrt{(1-\beta)(1-t_*)}\Bigl)=0.
\end{equation*}

The value of $t_*$ depends on $\beta$, and for the function $f(\beta,t)\equiv f(t)$ it holds $\dfrac{d}{d\beta}f(\beta,t_*(\beta))=
\dfrac{\partial}{\partial \beta}f(\beta, t_*(\beta))+\dfrac{dt_*(\beta)}{d\beta}\dfrac{\partial}{\partial t}f(\beta,t_*(\beta))=
\dfrac{\partial}{\partial \beta}f(\beta, t_*(\beta))$ since $t_*$ is critical for $f(t)$. Then,
\begin{eqnarray*}
\dfrac{d}{d\beta}f(t_*)=\dfrac{1}{\ln 2}\Bigl(\dfrac{t_*}{\beta}-\dfrac{1-t_*}{1-\beta}\Bigl)-\mu_0(p,\alpha)\Bigg(\sqrt{\dfrac{t_*}{\beta}}-\sqrt{\dfrac{1-t_*}{1-\beta}}\Bigg)\times \\
\times\Bigl(\sqrt{\beta t_*}+\sqrt{(1-\beta)(1-t_*)}\Bigl).
\end{eqnarray*}
Hence
\begin{equation*}
\ln 2\ \dfrac{d}{d\beta}f(t_*)=\Bigl(\dfrac{t_*}{\beta}-\dfrac{1-t_*}{1-\beta}\Bigl)+\sqrt{\dfrac{t_*(1-t_*)}{\beta(1-\beta)}}\Bigl(\ln\dfrac{1-t_*}{t_*}-\ln\dfrac{1-\beta}{\beta}\Bigl)\geq 0
\end{equation*}
since
\begin{equation*}
\sqrt{\dfrac{1-\beta}{\beta}\dfrac{t_*}{1-t_*}}-\sqrt{\dfrac{\beta}{1-\beta}\dfrac{1-t_*}{t_*}}+\ln\Bigl(\dfrac{\beta}{1-\beta}\dfrac{1-t_*}{t_*}\Bigl)= R-1/R-2\ln R\geq 0
\end{equation*}
where $R\equiv\sqrt{\dfrac{1-\beta}{\beta}\dfrac{t_*}{1-t_*}}\geq 1$.\\

It is simple to calculate that for $\beta=\alpha$, we have $t_*=p$, and $f(p)=0$. Therefore, $f(t_*)\geq 0$ for all $\alpha\leq \beta\leq 1/2$. We obtained the desired inequality
$f(t)\geq 0$ for the minimizer $t_*$ and hence for all the other points $t\in[0,1]$.
\end{enumerate}
\end{proof}

In the special case of \eqref{ac_Lem3_ineq3}, where $\alpha =1/2$, one has
an interesting inequality
\begin{equation*}
h_{1/2}(p)\equiv \sqrt{2p}\log (2p)+\sqrt{2(1-p)}\log (2(1-p))\geq 0,
\end{equation*}%
which takes place for $p\in (0,1)$. The function $h_{1/2}(p)$ is
nonnegative, convex and $h_{1/2}(p)\leq \sqrt{2}$.

\subsection{The case of theorem \protect\ref{th3}}

\label{sec:obtuse}

This section is devoted to proving the inequality \eqref{basic2} which
verifies the condition (ii.a) in the case of obtuse pyramid. Namely, \textit{%
if the parameter $p<1$ satisfies $p\geq\max\left\{\dfrac{m-1}{m}%
,p(m)\right\} $ with $p(m)$ defined in \eqref{cond}, then for complex
numbers $z_1,...,z_m$ satisfying (\ref{condw1}) the inequality
\begin{equation*}  \label{tildef}
\tilde{F}(z)\equiv -\sum\limits_{j=1}^m |z_j|^2\log|z_j|^2 - \tilde{\mu}%
_0(p)\Bigl|\sum\limits_{j=1}^m z_j\Bigl|^2+\tilde{\mu}_1(p)\geq 0.
\end{equation*}
holds with $\tilde{\mu}_0(p), \tilde{\mu}_1(p)$ defined in (\ref{0tilde}), (%
\ref{1tilde}).}
\begin{proof}[Suggestion for proof with computer verification]
As it has been established in lemma \ref{lemma}, it suffices to consider real vectors.

For real vectors $(z_j)_{j=1}^m$ on the unit sphere let us consider the Lagrange function
\begin{equation*}
L(z_1,...,z_k)=-\sum\limits_{k=1}^m z_k^2\log z_k^2-\tilde{\mu}_0(p)\Bigl(\sum\limits_{k=1}^m z_k\Bigl)^2+\tilde{\mu}_1(p)-\lambda\Bigl(\sum\limits_{k=1}^m z_k^2-1\Bigl).
\end{equation*}

The first-order conditions yield
\begin{equation*}
-\dfrac{1}{2}\dfrac{\partial L}{\partial z_k}=z_k\log z_k+\tilde{\mu}_0(p)\Bigl(\sum\limits_{j=1}^m z_j\Bigl)+(\lambda+\log e) z_k=0,\quad 1\leq k\leq m,
\end{equation*}
and the set $\{z_j\}$ has at most three different values $s_1,s_2,s_3$ with
multiplicities $k_1,k_2,k_3\geq 0$ ($k_1+k_2+k_3=m$),  since equation
\begin{equation}\label{T3_eq1}
z\log z^2+\tilde{\mu}_0(p)\Bigl(\sum\limits_{j=1}^m z_j\Bigl)+(\lambda+\log e) z=0
\end{equation}
in a variable $z$ has at most three real solutions. The fact that the equation
of the form $z\log z^2+az+b=0$ for real $a$ and $b$ has no more than three real roots was
actually used in the proof of lemma \ref{ac_Lem2} for equation \eqref{ac_Lem2_eq1} which is similar to \eqref{T3_eq1} with $\leq 2$ roots. However,
here we consider the entire real line, so the maximum number of real roots is three.

We consider two cases.

{\it Case 1. The minimizer has a zero component.}\\
Then $\lambda+\log e=0$ and the equation \eqref{T3_eq1} has roots $0$ and $\pm s$. Therefore the minimizer of
$\tilde{F} (z)$ must be of the form
$(\underbrace{s,-s,...,s,-s}_{k\text{ times}},0,...,0)$, where $2ks^2=1$ giving the value $-2ks^2\log s^2+\tilde{\mu}_1(p)=\log 2k+\tilde{\mu}_1(p)$
which has the minimum equal to $1+\tilde{\mu}_1(p)$ for $k=1$. The value $1+\tilde{\mu}_1(p)$ is nonnegative for the following reason.

First, we prove that the function $\tilde{\mu}_1(p)$ is increasing for $p$ in the range $\Bigl[\dfrac{m-1}{m},1\Bigr)$.
The derivative is

\begin{equation*}
\tilde{\mu}_1'(p)=\log e\dfrac{2(m-2)\sqrt{\dfrac{p(1-p)}{m-1}}-4p+2+\ln p-\ln \dfrac{1-p}{m-1}}{2(m-1)\sqrt{\dfrac{p(1-p)}{m-1}}\Bigg(\sqrt{p}+\sqrt{\dfrac{1-p}{m-1}}\Bigg)^2}.
\end{equation*}

Here, $2(m-2)\sqrt{\dfrac{p(1-p)}{m-1}}\geq 0$ and $\ln\dfrac{p}{1-p}-4p+2\geq 0$, hence $\tilde{\mu}_1'(p)\geq 0$.\\

Thus for $\frac{m-1}{m}\leq p< 1$ we have
\begin{equation*}
-\tilde{\mu}_1(p)\leq  -\tilde{\mu}_1\Bigl(\frac{m-1}{m}\Bigl)=-\dfrac{m-2}{m}\log(m-1)+\log m
\end{equation*}
by (\ref{val}). For $p(m)\leq p< 1$ we have
\begin{equation*}
-\tilde{\mu}_1(p)\leq -\tilde{\mu}_1(p(m))=1,
\end{equation*}
by (\ref{val1}). 
Therefore, $1+\tilde{\mu}_1(p)\geq 0$ for $p>\max(p(m),(m-1)/m)$. Note that
$\tilde{\mu}_1(p)\geq \tilde{\mu}_1((m-1)/m)\geq \tilde{\mu}_1(p(m))=-1$ if $p(m)<(m-1)/m$.

{\it Case 2. The minimizer has no zero components.}\\
The Lagrange multipliers method predicts that the minimizer has at most 3 different values,
that are in our case nonzero. They are denoted as $s_1<0<  s_2\leq s_3$
(without loss of generality we can assume that only $s_1$ is negative)
with multipliers $k_1,k_2,k_3$,
such that $k_1+k_2+k_3=m$ and $k_1s_1^2+k_2s_2^2+k_3s_3^2=1$. Our hypothesis is that,
however, there is only two distinct nonzero values among $s_1,s_2,s_3$. This fact can be
verified by computer experiments; we performed numerical simulations for
$m$ from $3$ to $50$ and for the parameter $a=\tilde{\mu}_0(p)>0$ taking values in $\{0.1, 0.5, 1, 5, 10, 100, \infty\}$ (see Appendix). 

Thus, up to permutation and an overall sign change, $(z_j)_{j=1}^m=(\underbrace{s_1,...,s_1}_{k_1\text{ times}},\underbrace{s_2,...,s_2}_{k_2\text{ times}})$, ${s_1< 0}$, ${s_2>0}$. The components of the minimizer must have opposite signs;
otherwise, changing the sign of a single component would decrease $\tilde{F}$.

Denote
\begin{eqnarray*}
\tilde{\mu}_0(p,\alpha)=\dfrac{\sqrt{\dfrac{p(1-p)}{\alpha(1-\alpha)}}\Bigl(\log\dfrac{p}{\alpha}-\log\dfrac{1-p}{1-\alpha}\Bigl)}{\Bigl(\sqrt{(1-p)(1-\alpha)}-\sqrt{p\alpha}\Bigl)\Bigl(\sqrt{\dfrac{p}{\alpha}}+\sqrt{\dfrac{1-p}{1-\alpha}}\Bigl)},\\
\tilde{\mu}_1(p,\alpha)=\dfrac{\sqrt{\dfrac{p}{\alpha}}\log\dfrac{p}{\alpha}+\sqrt{\dfrac{1-p}{1-\alpha}}\log\dfrac{1-p}{1-\alpha}}{\sqrt{\dfrac{p}{\alpha}}+\sqrt{\dfrac{1-p}{1-\alpha}}},
\end{eqnarray*}
where $\alpha$ and $p$ lie on $(0,1)$.\\

Relation with $\tilde{\mu}_0(p)$ and $\tilde{\mu}_1(p)$ is given by
\begin{eqnarray*}
\tilde{\mu}_0(p)=\dfrac{1}{m}\tilde{\mu}_0(p,1/m),\\
\tilde{\mu}_1(p)=\tilde{\mu}_1(p,1/m)-\log m.
\end{eqnarray*}

It is easy to see that for $\alpha$ and $p$ such that $0<\alpha\leq 1/2\leq1-\alpha\leq p<1$ one has  $\log\dfrac{p}{\alpha}-\log\dfrac{1-p}{1-\alpha}\geq 0$
and ${\sqrt{(1-p)(1-\alpha)}-\sqrt{p\alpha}\leq 0}$, therefore, $\tilde{\mu}_0(p,\alpha)\leq 0$.\\

The initial inequality \eqref{basic2} for $k_1 s_1^2=t$ is equivalent to
\begin{equation*}\label{obt_ineq1+}
-t\log\dfrac{t}{k_1/m}-(1-t)\log\dfrac{1-t}{k_2/m}\geq \tilde{\mu}_0(p,1/m)\Bigl(\sqrt{t(k_1/m)}-\sqrt{(1-t)(k_2/m)}\Bigl)^2-\tilde{\mu}_1(p,1/m),
\end{equation*}
and the result follows from lemma \ref{obt_Lem1}.

\end{proof}

\begin{Lem}
\label{obt_Lem1} Let $0\leq \alpha\leq \beta\leq \dfrac{1}{2}\leq
1-\alpha\leq p\leq 1$, $0\leq t\leq 1$. Then
\begin{equation*}
\tilde{f}(t)\equiv -t\log\dfrac{t}{\beta}-(1-t)\log\dfrac{1-t}{1-\beta}- %
\Bigg(\tilde{\mu}_0(p,\alpha)\Bigl(\sqrt{\beta t}-\sqrt{(1-\beta)(1-t)}%
\Bigl)^2-\tilde{\mu}_1(p,\alpha)\Bigg)\geq 0.
\end{equation*}
\end{Lem}

\begin{proof}
We will prove that the inequality
\begin{equation*}
\mu_0(p,\alpha)\Bigl(\sqrt{\beta t}+\sqrt{(1-\beta)(1-t)}\Bigl)^2-\mu_1(p,\alpha)\geq \tilde{\mu}_0(p,\alpha)\Bigl(\sqrt{\beta t}-\sqrt{(1-\beta)(1-t)}\Bigl)^2-\tilde{\mu}_1(p,\alpha)
\end{equation*}
holds in this range. So, the statement of the lemma follows from lemma \ref{ac_Lem3}.\\

Collecting similar terms and using notation of $f(t)$ from subsection \ref{prroof},
\begin{eqnarray*}
\tilde{f}(t)-f(t)= 2\sqrt{\beta(1-\beta)}\Bigl(\mu_0(p,\alpha)+\tilde{\mu}_0(p,\alpha)\Bigl)\sqrt{t(1-t)}-(1-2\beta)\Bigl(\mu_0(p,\alpha)-\tilde{\mu}_0(p,\alpha)\Bigl)t+\\
+\Bigl[(1-\beta)\Bigl(\mu_0(p,\alpha)-\tilde{\mu}_0(p,\alpha)\Bigl)-\Bigl(\mu_1(p,\alpha)-\tilde{\mu}_1(p,\alpha)\Bigl)\Bigl]\geq 0.
\end{eqnarray*}
One can see that
\begin{eqnarray*}
\mu_0(p,\alpha)-\tilde{\mu}_0(p,\alpha)=2(2p-1)\dfrac{\sqrt{p(1-p)\alpha(1-\alpha)}\Bigl(\log\dfrac{p}{\alpha}-\log\dfrac{1-p}{1-\alpha}\Bigl)}{(p-\alpha)(p+\alpha-1)},\\
\mu_0(p,\alpha)+\tilde{\mu}_0(p,\alpha)=-2(1-2\alpha)p(1-p)\dfrac{\log\dfrac{p}{\alpha}-\log\dfrac{1-p}{1-\alpha}}{(p-\alpha)(p+\alpha-1)}\leq 0,\\
\mu_1(p,\alpha)-\tilde{\mu}_1(p,\alpha)=2\dfrac{\sqrt{p(1-p)\alpha(1-\alpha)}\Bigl(\log\dfrac{p}{\alpha}-\log\dfrac{1-p}{1-\alpha}\Bigl)}{p-\alpha}.
\end{eqnarray*}

The function $(\tilde{f}-f)(t)$ has the form $A\sqrt{t(1-t)}+Bt+C$ with $A\leq 0$ and $B\leq 0$, then the unique minimum point $t_*\in[1/2,1]$.
It has the value $t_*=\dfrac{1}{2}\Bigl(1-\dfrac{B}{\sqrt{A^2+B^2}}\Bigl)$ and the minimum of the function $(\tilde{f}-f)(t)$ is equal to
\begin{equation*}
(\tilde{f}-f)(t_*)=\dfrac{B}{2}+C-\dfrac{1}{2}\sqrt{A^2+B^2},
\end{equation*}
where
\begin{eqnarray*}
A=2\sqrt{\beta(1-\beta)}\Bigl(\mu_0(p,\alpha)+\tilde{\mu}_0(p,\alpha)\Bigl),\\
B=-(1-2\beta)\Bigl(\mu_0(p,\alpha)-\tilde{\mu}_0(p,\alpha)\Bigl),\\
C=\Bigl[(1-\beta)\Bigl(\mu_0(p,\alpha)-\tilde{\mu}_0(p,\alpha)\Bigl)-\Bigl(\mu_1(p,\alpha)-\tilde{\mu}_1(p,\alpha)\Bigl)\Bigl].
\end{eqnarray*}

One has
\begin{eqnarray*}
\dfrac{B}{2}+C=\dfrac{\sqrt{p(1-p)\alpha(1-\alpha)}\Bigl(\log\dfrac{p}{\alpha}-\log\dfrac{1-p}{1-\alpha}\Bigl)}{(p-\alpha)(p+\alpha-1)}\times\\
\times\Bigl((2p-1)-2(p+\alpha-1)\Bigl)=\\
=(1-2\alpha)\dfrac{\sqrt{p(1-p)\alpha(1-\alpha)}\Bigl(\log\dfrac{p}{\alpha}-\log\dfrac{1-p}{1-\alpha}\Bigl)}{(p-\alpha)(p+\alpha-1)}\geq 0,\\
\dfrac{1}{2}\sqrt{A^2+B^2}=(1-2\alpha)\dfrac{\sqrt{p(1-p)\alpha(1-\alpha)}\Bigl(\log\dfrac{p}{\alpha}-\log\dfrac{1-p}{1-\alpha}\Bigl)}{(p-\alpha)(p+\alpha-1)}\times\\
\times \sqrt{4\ \dfrac{p(1-p)\beta(1-\beta)}{\alpha(1-\alpha)}+\dfrac{(1-2p)^2(1-2\beta)^2}{(1-2\alpha)^2}}
\end{eqnarray*}
Thus,
\begin{eqnarray}\label{obtL2_ineq1}
\dfrac{\sqrt{A^2+B^2}}{B+2C}=\sqrt{4\ \dfrac{p(1-p)\beta(1-\beta)}{\alpha(1-\alpha)}+\dfrac{(1-2p)^2(1-2\beta)^2}{(1-2\alpha)^2}}\leq 1
\end{eqnarray}

Let us denote $u(t)=4t(1-t)$. The inequality \eqref{obtL2_ineq1} can be rewritten in the form
\begin{equation*}
\dfrac{u(\beta)u(p)}{u(\alpha)}+\dfrac{(1-u(\beta))(1-u(p))}{1-u(\alpha)}\leq 1,
\end{equation*}
where $u(p)\leq u(\alpha)\leq u(\beta)$. The linear function
\begin{equation*}
t\mapsto W(t)=\dfrac{u(\beta)t}{u(\alpha)}+\dfrac{(1-u(\beta))(1-t)}{1-u(\alpha)}
\end{equation*}
has non-negative derivative $\dfrac{u(\beta)}{u(\alpha)}-\dfrac{1-u(\beta)}{1-u(\alpha)}=\dfrac{u(\beta)-u(\alpha)}{u(\alpha)(1-u(\alpha))}$ and ${W(u(\alpha))=1}$.
Therefore, the inequality \eqref{obtL2_ineq1} holds, and lemma has been proven.
\end{proof}

\section*{Appendix. Numerical Verification of Hypotheses for Minimization on the Sphere (Generated by DeepSeek)}\label{Appendix}

We consider the minimization problem of the function
\begin{equation}\label{tildeF}
\tilde{F}(z_1,\dots,z_m) = -\sum_{k=1}^m z_k^2\log_2(z_k^2) + a\left(\sum_{k=1}^m z_k\right)^2
\end{equation}
on the unit sphere $\sum_{k=1}^m z_k^2 = 1$, where the parameter $a>0$
equals to $\tilde{\mu}_0(p)$ given by \eqref{0tilde} in the obtuse scenario; in the flat scenario (formally corresponding to the case
$a=\infty$) the second term in the righthand side of \eqref{tildeF} is absent and there is  additional constraint $\sum_{k=1}^m z_k = 0$.

The main hypothesis suggests that minimizers have at most two distinct nonzero values.
To verify this hypothesis numerically, we implement the following verification scheme:

\begin{algorithm}[H]
\caption{Numerical verification of the minimization hypothesis}
\begin{algorithmic}[1]

\State \textbf{Input:} Parameter $a$, dimension range $[m_{\min}, m_{\max}]$, tolerance $\epsilon = 10^{-8}$
\State \textbf{Output:} Verification result for each dimension $m$

\Procedure{VerifyHypothesis}{$a$, $m_{\min}$, $m_{\max}$}
    \For{$m = m_{\min}$ \textbf{to} $m_{\max}$}
        \State Find approximate minimizer $\mathbf{z}^*$ on the sphere
        \State Compute $\tilde{F}_{\min} = \tilde{F}(\mathbf{z}^*)$

        \State \Comment{Test Hypothesis 1: Two-value structure}
        \State $\mathbf{z}_{\text{proj}} \gets \text{ProjectToTwoValueStructure}(\mathbf{z}^*, a)$
        \If{$\mathbf{z}_{\text{proj}} \neq \text{None}$}
            \State $\Delta \tilde{F}_1 \gets \tilde{F}(\mathbf{z}_{\text{proj}}) - \tilde{F}_{\min}$
            \State $\text{valid}_1 \gets (|\Delta \tilde{F}_1| \leq \epsilon)$
        \Else
            \State $\text{valid}_1 \gets \text{False}$
        \EndIf

        \State \Comment{Test Hypothesis 2: Special two-component vector}
        \State $\mathbf{z}_{\text{special}} \gets (1/\sqrt{2}, 0,\dots,0,-1/\sqrt{2})$
        \State $\Delta \tilde{F}_2 \gets \tilde{F}(\mathbf{z}_{\text{special}}) - \tilde{F}_{\min}$
        \State $\text{valid}_2 \gets (|\Delta \tilde{F}_2| \leq \epsilon)$

        \State \Comment{Accept hypothesis if either projection is nearly optimal}
        \State $\text{confirmed} \gets \text{valid}_1 \ \textbf{or} \ \text{valid}_2$
        \State Output result for dimension $m$
    \EndFor
\EndProcedure

\end{algorithmic}
\end{algorithm}

For a given vector $\mathbf{z}^*$, we compute the averages $u$ and $v$
of its positive and negative components, respectively, then rescale them
to lie on the unit sphere. In the case $a=\infty$, the exact values
$u = \sqrt{\ell/(k m)}$, $v = -\sqrt{k/(\ell m)}$ are used to satisfy
both $\|\mathbf{z}\|=1$ and $\sum_i z_i=0$. The projection replaces all
positive entries by $u$ and all negative entries by $v$, preserving zeros.

\subsection*{Verification Criterion}

The hypothesis is considered confirmed for dimension $m$ if either:
\begin{itemize}
    \item The projection $\mathbf{z}_{\text{proj}}$ onto a two-value structure yields a function value within tolerance $\epsilon$ of the computed minimum: $|\tilde{F}(\mathbf{z}_{\text{proj}}) - \tilde{F}_{\min}| \leq \epsilon$
    \item The special two-component vector $\mathbf{z}_{\text{special}} = (1/\sqrt{2},0,\dots,0,-1/\sqrt{2})$ yields a function value within tolerance $\epsilon$ of the computed minimum
\end{itemize}

We use $\epsilon = 10^{-8}$ throughout our computations. The optimization is performed using SciPy's SLSQP algorithm with multiple initial points to ensure robustness.

\subsection*{}
\begin{alltt}
======================================================================
VERIFICATION OF THE MINIMIZATION HYPOTHESIS FOR FUNCTION F ON THE SPHERE
======================================================================
Enter parameter a: 10
Enter m_min: 3
Enter m_max: 50

======================================================================
Parameters: a = 10, m in [3, 50]
Criterion: hypothesis holds if \(\Delta\tilde{F}\)_1 < 1e-8 or \(\Delta\tilde{F}\)_2 < 1e-8
======================================================================

   m |    \(\tilde{F}\)_min    |    \(\Delta\tilde{F}\)_1     |    \(\Delta\tilde{F}\)_2     | Result
-----------------------------------------------------------------
   3 |   1.00000000 |     2.52e-01 |    -9.44e-15 |    YES
   4 |   1.00000000 |     2.08e-01 |    -4.11e-15 |    YES
   5 |   1.00000000 |     1.22e-01 |    -4.33e-15 |    YES
   6 |   0.98734793 |    -1.44e-15 |     1.27e-02 |    YES
   7 |   0.91479512 |     0.00e+00 |     8.52e-02 |    YES
   8 |   0.85213438 |     1.11e-16 |     1.48e-01 |    YES
   9 |   0.79788958 |    -8.88e-16 |     2.02e-01 |    YES
  10 |   0.75061687 |     1.11e-16 |     2.49e-01 |    YES
  11 |   0.70910123 |    -3.33e-16 |     2.91e-01 |    YES
  12 |   0.67236219 |    -3.33e-16 |     3.28e-01 |    YES
  13 |   0.63961653 |    -3.33e-16 |     3.60e-01 |    YES
  14 |   0.61023775 |     0.00e+00 |     3.90e-01 |    YES
  15 |   0.58372158 |    -2.22e-16 |     4.16e-01 |    YES
  16 |   0.55965888 |    -2.22e-16 |     4.40e-01 |    YES
  17 |   0.53771482 |    -1.67e-15 |     4.62e-01 |    YES
  18 |   0.51761305 |     2.22e-16 |     4.82e-01 |    YES
  19 |   0.49912361 |     5.55e-17 |     5.01e-01 |    YES
  20 |   0.48205361 |    -1.11e-16 |     5.18e-01 |    YES
  21 |   0.46624010 |    -5.55e-17 |     5.34e-01 |    YES
  22 |   0.45154446 |    -7.77e-15 |     5.48e-01 |    YES
  23 |   0.43784802 |    -2.78e-16 |     5.62e-01 |    YES
  24 |   0.42504857 |    -2.50e-15 |     5.75e-01 |    YES
  25 |   0.41305761 |    -2.78e-16 |     5.87e-01 |    YES
  26 |   0.40179807 |    -8.88e-16 |     5.98e-01 |    YES
  27 |   0.39120256 |    -3.33e-16 |     6.09e-01 |    YES
  28 |   0.38121185 |    -2.22e-16 |     6.19e-01 |    YES
  29 |   0.37177368 |    -2.33e-15 |     6.28e-01 |    YES
  30 |   0.36284175 |     1.11e-16 |     6.37e-01 |    YES
  31 |   0.35437487 |     3.89e-16 |     6.46e-01 |    YES
  32 |   0.34633629 |    -2.78e-16 |     6.54e-01 |    YES
  33 |   0.33869310 |    -1.11e-16 |     6.61e-01 |    YES
  34 |   0.33141576 |    -5.55e-17 |     6.69e-01 |    YES
  35 |   0.32447761 |    -2.05e-15 |     6.76e-01 |    YES
  36 |   0.31785461 |     3.33e-16 |     6.82e-01 |    YES
  37 |   0.31152497 |    -2.66e-15 |     6.88e-01 |    YES
  38 |   0.30546889 |    -7.61e-15 |     6.95e-01 |    YES
  39 |   0.29966835 |    -3.33e-16 |     7.00e-01 |    YES
  40 |   0.29410691 |    -1.59e-14 |     7.06e-01 |    YES
  41 |   0.28876954 |    -6.11e-16 |     7.11e-01 |    YES
  42 |   0.28364247 |    -2.61e-15 |     7.16e-01 |    YES
  43 |   0.27871305 |    -1.67e-16 |     7.21e-01 |    YES
  44 |   0.27396965 |    -8.60e-15 |     7.26e-01 |    YES
  45 |   0.26940154 |    -8.60e-15 |     7.31e-01 |    YES
  46 |   0.26499883 |     3.33e-16 |     7.35e-01 |    YES
  47 |   0.26075237 |     2.78e-16 |     7.39e-01 |    YES
  48 |   0.25665367 |    -4.00e-15 |     7.43e-01 |    YES
  49 |   0.25269489 |     5.55e-16 |     7.47e-01 |    YES
  50 |   0.24886871 |    -1.67e-16 |     7.51e-01 |    YES

======================================================================
FINAL RESULTS:
======================================================================
 YES: Hypothesis CONFIRMED for all 48 values of m

\end{alltt}

The computational results reveal that for large $a$ and small $m$ ($< 7$),
the minimizer typically contains only two nonzero entries of opposite sign.
Otherwise, all components are nonzero and assume exactly two distinct values.

\section{Postscriptum (added 15/07/2026)}

In response to our  postings \verb"arXiv:2506.06700" (first version of this one, related to the inequalities for the Shannon entropy) and \verb"arXiv:2603.24017" (dealing with the corresponding inequalities for the R\'enyi entropy in the flat case), the preprints of Haonan Zhang, \verb"arXiv:2605.05243",  and of Alvan Arulandu, arXiv:2606.14698, appeared, suggesting independent AI-assisted proofs of our conjectured inequalities for flat and, respectively, obtuse pyramids. Further, we wish to add that inequality \eqref{hacute} for strongly acute pyramid is closely related to the sharp log-Sobolev inequality for simple random walk on $Z_{d}$ (Pott's model) computed in P. Diaconis, L. Saloff-Coste, Logarithmic Sobolev inequalities for finite Markov chains, The Annals of Applied Probability, \textbf{6} (3), 695-750, (1996). Thus all our conjectured tight inequalities, along with their consequences for the accessible information of quantum pyramids, are completely analytically confirmed.

Here we also wish to make a clarification concerning the general optimality conditions in theorem \ref{th1}. The relations \textbf{iii.a,b} of this theorem give the universally valid necessary and sufficient conditions for the optimality; however the relations \textbf{ii.a,b} are applicable, strictly speaking, only under certain additional assumption. Namely, the quantities $\mathrm{Tr\,}\,\sigma
M_{j}^{\prime }$ in the equation \eqref{K} defining $K(\sigma)$ should be all positive in order to avoid the infinite value of the logarithms. In particular, this holds for all $\sigma$ if all the operators $M_{j}^{\prime }$ are strictly positive. The conditions \textbf{ii.a,b} then acquire literal meaning. On the other hand, such an assumption does not hold for an ensemble of pure linearly independent states like quantum pyramids. However, while the inequality \textbf{{ii.a}} remains valid \textit{formally}, tolerating the value $+\infty$ in the righthand side, the equation \textbf{ii.b} which was used to find the eigenvalues of operator $\Lambda_0$ for this ensemble, is still valid literally because the expression
\begin{equation*}
K(\sigma _{k}^{0})\,|\phi_{k}^{0}\rangle = -\sum_{j}M_{j}^{\prime }|\phi
_{k}^{0}\rangle \log \langle \phi _{k}^{0}|M_{j}^{\prime }|\phi
_{k}^{0}\rangle
\end{equation*}
is always well defined taking into account the usual convention $0\log 0=0$.
Therefore the derivation of the operator $\Lambda_0$ and of all the corresponding hypothetical entropy inequalities
also remains valid. In fact, corresponding calculations could be performed on the basis of the more appropriate equation \textbf{iii.b}, which however would be less intuitive. Alternative proof of the conditions \textbf{iii.a,b} avoiding introduction of the operators $K(\sigma)$ is outlined in the paper: A.S. Holevo, Probability measures on the set of quantum states
in optimization problems for information quantities, Theory Probab. Appl. \textbf{71}(2), (2026).

\bigskip
\centerline{Funding}
This work is supported by Russian Science
Foundation under the grant No 24-11-00145.


\begin{thebibliography}{99}
\bibitem{dav} E. B. Davies, Information and quantum measurement, IEEE Trans.
Inf. Theory \textbf{24} (1978) 596-599.

\bibitem{eng} B.-G. Englert and J. \v{R}eh\'{a}\v{c}ek, How well can you
know the edge of a quantum pyramid, J. Mod. Optics \textbf{57 }N3 (2010)
218-226. Arxiv:0905.0510.

\bibitem{gross} L. Gross, Logarithmic Sobolev inequalities, Amer. J. Math.
\textbf{97}, 1061--1083 (1975).

\bibitem{hel} C. W. Helstrom, Quantum Detection and Estimation Theory,
(Academic Press, New York 1976).

\bibitem{ippi} A. S. Holevo, Informational aspects of quantum measurement,
Probl. Peredachi Inf. \textbf{9} (1973) 31-42; English translation: Probl.
Inf. Transm. (USSR) \textbf{9} (1973) 110-118.

\bibitem{jma1} A. S. Holevo, Statistical decision theory for quantum
systems, J. Multivariate Anal., \textbf{3} (1973), 337--394.


\bibitem{ljm} A. S. Holevo, On Optimization Problem for Positive
Operator-Valued Measures, Lobachevskii J. Math., \textbf{43}:7 (2022),
1646-1650.

\bibitem{ljm1} A. S. Holevo, An optimization problem concerning noise in
quantum measurement channels, Lobachevskii J. Math., \textbf{44}:6 (2023),
2033-2043.

\bibitem{jozsa} R. Jozsa, D. Robb, W. K. Wootters, Lower bound for
accessible information in quantum mechanics, Phys. Rev. A, \textbf{49}:2
(1994), 668-677.

\bibitem{keil} A. Keil, Proof of the orthogonal measurement conjecture for
qubit states, Continuity of the relative entropy of resource, IJQI Special
Issue Alexander Semenovich Holevo's 80th Birthday, \textbf{22} N5 2440008
(2024). arXiv:0809.0232

\bibitem{hkron} A. S. Holevo, D. A. Kronberg, New approaches to eavesdropper
information bounds in quantum cryptography problems, Voprosy
kiberbezopasnosti, 3(67) (2025), 110-117 (In Russian).

\bibitem{lev} L. B. Levitin, Optimal quantum measurements for two pure and
mixed states, in: Quantum Communications and Measurement, edited by V. P.
Belavkin, O. Hirota, and R. L. Hudson, (Plenum Press, New York 1995) 439-448.

\bibitem{moser} S. M. Moser, Information Theory (Lecture Notes), 6th
edition, ETH Z\"urich and NYCU, 2018.

\bibitem{sas} M. Sasaki, S. M. Barnett, R. Jozsa, M. Osaki, and O. Hirota,
Accessible information and optimal strategies for real symmetrical quantum
sources, Phys. Rev. \textbf{A 59} (1999) 3325-3335. arXiv:quant-ph/9812062

\bibitem{shor} P. W. Shor, On the Number of Elements in a POVM Attaining the
Accessible Information, arXiv:quant-ph/0009077 (2000).

\bibitem{obzor} J. Suzuki, S. M. Assad and B.-G. Englert, Accessible
information about quantum states: An open optimization problem, in
Mathematics of Quantum Computation and Quantum Technology, eds. G. Chen, L.
Kauffman and S. J. Lomonaco (Chapman \& Hall, 2007), pp. 309-348.


\end{thebibliography}
\end{document}